\documentclass[preprint,12pt]{elsarticle}

\usepackage{amssymb}
\usepackage{amsmath}
\usepackage{amsthm}
\usepackage{graphicx}
\usepackage{amsmath}
\usepackage{amssymb}
\usepackage{booktabs}
\usepackage{array}
\usepackage{float} 
\usepackage{tabularx} 
\usepackage{longtable}
\usepackage{array}
\usepackage{multirow}
\usepackage{caption} 
\renewcommand{\arraystretch}{1.2} 
\usepackage{algorithm}
\usepackage{algorithmic}
\usepackage{longtable}
\usepackage{multicol} 
\usepackage{array}
\usepackage{hyperref}
\usepackage{makecell}
\usepackage{subcaption}
\newtheorem{proposition}{Proposition}
\newtheorem{theorem}{Theorem}

\newtheorem{definition}{Definition}
\newtheorem{example}[theorem]{Example}
\usepackage{float}   
\usepackage{paracol}
\usepackage{longtable}
\usepackage{caption}
\usepackage{csquotes}
\usepackage{multicol}
\usepackage{color}
\usepackage{longtable}
\usepackage{pdflscape}   
\usepackage{booktabs}    
\usepackage{enumerate,enumitem}
\begin{document}

\begin{frontmatter}

\title{{Cellular Automata based Resource Efficient Maximally Equidistributed Pseudo-Random Number Generators} }




\author{Bhuvaneswari A}
\ead{abhuvanesh15@gmail.com}
\affiliation{organization={Department of Computer Science and Engineering, National Institute of Technology},
            city={Tiruchirappalli},
            postcode={620015}, 
            state={Tamil Nadu},
            country={India}}
\author{Kamalika Bhattacharjee\corref{mycorrespondingauthor}}
\cortext[mycorrespondingauthor]{Corresponding author}
 \ead{kamalika.it@gmail.com}
\affiliation{organization={Department of Information Technology, Indian Institute of Engineering Science and Technology, Shibpur},
            city={Howrah},
            postcode={711103}, 
            state={West Bengal},
            country={India}}



\vspace{-2em}
\begin{abstract}
 An equidistribution is a theoretical quality criteria that measures the uniformity of a linear pseudo-random number generator (PRNG). In this work, we first show that all existing linear cellular automaton (CA) based pseudo-random number generators (PRNGs) are weak in the equidistribution characteristic. Then we propose a list of light-weight combined CA-based PRNGs with time spacing ($2 \leq s \leq 10$) using linear maximal length cellular automata of degree $31 \leq k \leq 128$ (close to computer word size). We show that these PRNGs achieve maximal period as well as satisfy the maximal equidistribution property. Finally, we show that these combined maximal length CA-based PRNGs pass almost all the empirical testbeds, with speed and performance comparable to the Mersenne Twister.

\end{abstract}


\end{frontmatter}

\section{Introduction}
\label{section1}
Pseudo-random Number Generators (PRNGs) are designed to produce sequences of random numbers with long period length, high efficiency, and good theoretical quality. Because of these characteristics, they are widely used in various applications such as simulations, cryptography, and gambling, etc.
A PRNG is a deterministic algorithm that starts with a random seed value and produces a sequence of numbers, known as the pseudo-random numbers. The primary characteristics of a PRNG are large period, reproducibility, uniformity, and efficiency. In addition to these characteristics, certain theoretical quality merits need to be satisfied, such as \emph{equidistribution} that measures the uniformity in each of the possible dimensions
\cite{F26}. The PRNGs should also pass all or the majority of benchmark empirical testbeds that determine the randomness quality.

Majority of the current-day PRNGs are based on the linear feedback shift registers (LFSRs) where there are underlying \emph{primitive polynomials} (modulo 2) responsible for the large period and bitwise implementation. Numerous PRNGs exist based on LFSRs such as Tausworthe \cite{taus9}, Mersenne Twister \cite{mersenne19}, GFSR and Twisted GFSR \cite{tgfsr20}, WELL \cite{well18}, and Xorshift generators \cite{xor26}, etc., that produce the sequence of pseudo-random numbers based on linear recurrence modulo 2. These are referred to as linear PRNGs. These types of PRNGs are implemented easily, can be characterized by matrix algebra and the uniformity property can also be measured by equidistribution. The classical LFSR-based PRNG is called a Tausworthe generator. It generates the random numbers (Hereafter, the term \enquote{random} shall denote \enquote{pseudo-random} only) efficiently when the polynomial contains a smaller number of non-zero coefficients (\(N_1\)). However, it fails to satisfy the equidistribution property and many basic statistical tests. 
Also, since the polynomial contains fewer \(N_1\), only a small percentage of the bits are changed at each step, remaining all dominated by zeros throughout several cycles. In cryptography, it is known as low diffusion capacity problem \cite{F26}, which is undesirable for randomness. The reason for taking such polynomials is the lack of primitive polynomials of large degree that satisfy this high diffusion property.

To address these limitations, two approaches have been taken in the literature. First, use a combined generator \cite{PRNG2} which combines multiple pseudo-random sequences generated by multiple primitive polynomials of relatively smaller degree using the exclusive OR (XOR) operation. This type of combined PRNGs with carefully selected parameters not only achieves good \emph{theoretical figures of merit}, but it also produces a long period length suitable for randomness. Another approach is to use \emph{tempering} -- an output transformation function -- to manipulate the bits generated by the large degree linear recurrence to achieve better equidistribution \cite{combined14}.
However, both of these approaches work well when the component(s) linear recurrence(s) is (are) of sufficiently large. Whereas, several modern-day applications require \emph{light-weight} PRNGs suitable to be implemented inside FPGA platform in hardware as well as to be efficiently run and fit inside the word size of a multicore CPU or GPUs. These light-weight PRNGs have restriction on resource utilization, which makes them vulnerable to security flaws \cite{LightweightFlaw2010} resulting from poor randomness quality. The latest light-weight PRNGs target to exploit the chaotic properties of logistic maps, etc. \cite{lightweightChaotic2020, LightweightLogistic2025}, but that makes them useless to be tested for theoretical \emph{figures of merit}. 

On this backdrop, Cellular Automaton (CA), a natural model of computation and long acclaimed as a source of randomness in physical system \cite{ca21}, comes as a great alternative. A CA uses simple local rules to update the state of each cell parallely, which results in global dynamics. Because of the locality, inherent parallelism and simplicity, Cellular Automata (CAs) have been very popular for hardware implementation \cite{maxCA33}. These benefits have prompted researchers to use CAs as PRNGs \cite{PRNG1}. For this purpose, chaotic \emph{maximal length} CAs, a variant of the simplest form of the CAs, known as Elementary Cellular Automata (ECAs) have mostly been used \cite{maxCA33,parallel24,linear16}. Similar to LFSRs, for maximal length CAs also, there exist underlying primitive polynomials modulo 2 responsible for a period length $2^{k}-1$ where $k$ is the size of the CA. These CAs have also been utilized for designing light-weight and multiple stream PRNGs \cite{parallel24,parallel25}. However, although these PRNGs are very much resource efficient, they still lack in the randomness quality as measured by the benchmark statistical tests \cite{parallel24}. Further, theoretical figures of merit for these CA-based PRNGs are also unknown.

So, in this work, our target is to design resource-efficient PRNGs using the maximal length CAs. Ref. \cite{CA4} and \cite{CA7} provide good resources to get maximal length CAs, which are to be utilized in this work. To obtain larger period and maximal equidistribution, more than one maximal length CAs will be combined in the PRNGs. To make it light-weight, each component CA is restricted to be less than or equal to 128 bits and the number of components is restricted to two. All these maximal length CAs are chaotic and efficient in terms of both hardware and memory requirement. But, being chaotic, these component CAs have self-similar patterns. To break these patterns and systematically improve the randomness quality, we use time-spacing.
We show that by using these CAs we can build PRNGs which are resource efficient, fast, having large period, yet are maximally equidistributed and pass the majority of tests in all benchmark testbeds.



\section{PRNGs and Quality Criteria}
\label{section2.1}
A pseudo-random number generator (PRNG) can be mathematically defined as follows \cite{F26}:
\begin{definition}
PRNG $G$ is a structure that consists of \( (S, \mu, f, U, g) \) : where \( S \) is a finite collection of states,  \( \mu \) is the probability distribution on \( S \) for the starting state, known as the seed, \( f: S \rightarrow S \) is the transition function to update state, \( U \) is the output space, and \( g: S \rightarrow U \) is the output function.  
\end{definition}

Based on $\mu$, an initial state $s_0 \in S$ is selected as seed. Initially, the output is $u_0 = g(s_0)$. The new state is calculated as $s_i = f(s_{i-1})$ for each step $i \geq 1$, and the associated output is $u_i = g(s_i)$. This sequence $(u_i)_{i \geq 0}$ is referred to as the pseudo-random sequence produced by $G$.

\subsection{General framework of linear PRNG}
\label{section2.2}
The general framework of the linear PRNG based on the matrix linear recurrence over $\mathbb{F}_2$ is represented by the following equations \cite{F26}. 
\begin{align}\label{eq:prng}
x_n &= A x_{n-1} \\
y_n &= B x_{n} \\
u_n &= \sum_{l=1}^w y_{n,l-1} 2^{-l} = .\, y_{n,0} y_{n,1} y_{n,2}...
\end{align}

In those equations, $x_n$ is the $k$-bit state vector at step $n$ and $y_n$ is the $w$-bit output vector at step $n$. Both $k$ and $w$ are positive integers. $A$ is the \(  k\times k \)  transition or \emph{characteristic} matrix and $B$ is the \(w\times k \) output transformation or \emph{tempering} matrix. At step $n$, the output random number obtained is the real number \( u_n \in [0,1) \). Any linear PRNG can be fitted into this framework by selecting the proper matrices $A$ and $B$. 

An alternative representation of this sequence is given in Ref.~\cite{period5}. Here, a pseudo-random sequence $u_i \in [0,1), i = 1, 2, \cdots$ is defined as:
\begin{align}\label{eq:laurent}
    u_i = \sigma({f_i(x)}/{M(x)}) & \text{ where }
    f_i(x) = (g(x)f_{i-1}(x) +h(x)) \pmod{M(x)}
\end{align}
Here, $g(x), h(x), M(x)$ and $f_i(x)$ are polynomials in GF\{2, x\} and $\sigma$ is a mapping from GF\{2, x\} to the real filed. In case of Tausworthe generators, $M(x)$ is the characteristics polynomial of degree $k$ which is primitive modulo two, $h(x) = 0$, $g(x) = (x^s \pmod{M(x)})$, with $0 < s < 2^k - 1$ and $\gcd( s, 2^k- 1) = 1$. Such polynomial representation is often helpful to theorize the characteristics of the PRNG.

\subsection{Combined PRNG}
\label{section2.3}
A combined PRNG combines more than one distinct recurrence (component) of the form Equation \ref{eq:prng} as follows. Suppose we have the $j$ components, $j = 1,2,\ldots, J$. For the generator $j$, let $A_j$ be the $k_j \times k_j$ transition matrix and $B_j$ be the $w \times k_j$ output transformation matrix. The $k_j$ bit state vector at step $n$ is $x_{j,n}$. The output of step $n$ of this combined generator is defined by the following equations \cite{combined14}:
\begin{align}\label{eq:combinedprng}
y_n &= B_1 x_{1,n} \oplus B_2 x_{2,n} \oplus \cdots \oplus B_j x_{j,n} \\
u_n &= \sum_{\ell=1}^w y_{n,\ell}\, 2^{-\ell}
\end{align}

Here, $\oplus$ represents the bitwise exclusive-or (XOR) operation. The combined generator is similar to the generator given by \autoref{eq:prng}, with parameters $k = k_1 + \cdots + k_j$, $A = \operatorname{diag}(A_1, \ldots, A_J)$, and $B = (B_1, \ldots, B_J)$. Let $P_j(z)$ represent the characteristic polynomial of $A_j$ for each component $j$. Therefore, the characteristic polynomial of $A$ is $P(z) = P_1(z) \cdots P_J(z)$. The combined PRNG cannot achieve the \emph{maximal period} of $2^k - 1$ since this polynomial is reducible.
The choice of which components to combine is crucial, because, depending on the component selection, it results in a large period and better theoretical properties. Suppose we choose the parameters such that each component $i$ with degree $k_i$ has a maximal period $2^{k_i}-1$ and the period lengths of the chosen components are relatively prime. Then, the period of this combined PRNG is the least common multiple (LCM) of the periods of the individual components 
, as stated in Proposition \ref{prop1} \cite{period5}.

\begin{proposition} [\cite{period5}]\label{prop1}
A combined pseudo-random sequence defined as
$
U_n = u_n^{(1)} \text{ XOR } \cdots \text{ XOR } u_n^{(J)}
$, $n = 1, 2, \ldots$
has a period of  $\prod_{j=1}^J \left( 2^{\deg(M^{(j)})} - 1 \right)$ if for each sequence $u_n^{(j)}$, with characteristic polynomial $M^{(j)}$, the periods of the sequences, $2^{\deg(M^{(j)})} - 1$, $j = 1, \ldots, J$, are pairwise co-prime. 
\end{proposition}

\subsection{Theoretical Figures of Merit}
\label{section2.4}
The most important theoretical property of a PRNG is its generated numbers are uniformly distributed over any dimension. This uniformity of the output values can be evaluated using a procedure called \emph{equidistribution}. Consider the generator's output values as $\Psi_t = \{ (u_0, \dots, u_{t-1}):x_0 \in \mathbb{F}_2 \}$. The uniformity of a $\mathbb{F}_2$-linear generator is evaluated by using $\Psi_I = \{ (u_{i_1}, \dots, u_{i_t}): x_0 \in \mathbb{F}_2 \}$, where $I = \{ i_1, \dots, i_t \}$ is a fixed ordered group of non-negative integers with $0 \le i_1 < \mathbb{F}_2$ \cite{F26}. 
For $I = \{ 0, \dots, t-1 \}$, the initial set $\Psi_t = \Psi_I$ is returned. 

Let \(\mathbf{q} = (q_1, \ldots, q_t)\) be an arbitrary vector of non-negative integers. 
Divide the unit hypercube \([0, 1)^t\) into \(2^{q_j}\) equal-length intervals along each axis \(j = 1, \ldots, t\). This partitions \([0, 1)^t\) into \(2^{q_1 + \cdots + q_t}\) rectangular boxes of equal size and shape. If a set \(\Psi_I\) contains exactly \(2^q\) points in each box, where \(q\) 
satisfies $k - q = q_1 + \cdots + q_t,$ then \(\Psi_I\) is said to be \(\mathbf{q}\)-equidistributed.  This is possible only if \(q_1 + \cdots + q_t \leq k\). This \(\mathbf{q}\)-equidistribution can be used to measure the uniformity of the linear PRNG. It is verified by constructing a corresponding binary matrix and checking its rank.
The following steps are used to generate the binary matrix ($\mathcal{B}$) \cite{F26}. 

\begin{itemize}[noitemsep,topsep=0pt]
\item Let \(j=1,2,…k\)
\item Begin the generator for \(j=1,2,…k\) assuming that $x_0= e_j$ is the initial state and that $e_j$ is the unit vector, which has a 1 at position $j$ and zeros
\item Run the generator for $t$ time steps for every $j$ value 
\item Obtain the output $u_n$ from every step 
\item Extract the $q_1$ most significant bits of the output at step 1, $q_2$ at step 2, and $q_t$ at step $t$ 
\item These extracted bits make up the $j^{th} $column of the binary matrix. Repeat this for every column \(j \) 
\end{itemize}

\subsubsection{\((t,\ell)\)-equidistribution}
For \(\mathbf{q}\)-equidistribution, if the value of \(k\) is large, it is very difficult to compute. In such instances, take only a smaller class of vectors \(\mathbf{q}\): those for which all coordinates \(q_j\) are equal to a certain constant \(l \geq 1\). That is, consider only the \(l\) most significant bits of each coordinate, which divides the unit hypercube \([0,1)^t\) into \(2^{tl}\) cubic boxes. If \(\Psi_I\) is \(\mathbf{q}\)-equidistributed for \(\mathbf{q} = (l, \ldots, l)\), then \(\Psi_I\) is said to be \(t\)-distributed with \(l\)-bit accuracy. This \((t,l)\)-equidistribution can be verified using the following Proposition \ref{prop2} for individual generators \cite{{ME3}}. 

\begin{proposition} [\cite{ME3}]\label{prop2}
The sequence is \((t,\ell)\)-equidistributed if and only if the matrix 
\(\mathcal{B}_{t,\ell,s}\) has (full) rank \(t\ell\). 
If \(\ell_t = \lfloor k/t \rfloor < k/t \leq L\), 
then the sequence is also collision-free \(CF(t)\) if and only if the matrix 
\(\mathcal{B}_{t,\ell_t+1,s}\) has rank \(k\).
\end{proposition}

The maximal value of \(l\) for which this characteristic holds is termed as the \emph{resolution} of \(\Psi_I\), and is denoted by \(l_I\). It cannot go beyond \(l_t^{*} = \min(L, \lfloor k/t \rfloor)\). The \emph{resolution gap} for \(\Psi_I\) is defined as \(\Lambda_I= l_t^{*} - l_I\). The resolution gap in dimension \(t\) is defined as \(\Lambda_t = \ell_t^{*} - \ell_t\), and the dimension gap in resolution \(\ell\) is defined as \(\Delta_\ell = t_\ell^{*} - t_\ell\).

These concepts apply not just to irreducible polynomials but also to reducible polynomials (combined generators) \cite{ME3}. In the case of a combined generator, compute the binary matrix for every component. Then, combining the matrices using vertical juxtaposition, we got an equidistribution matrix (\(\widetilde{\mathcal{B}}_{t, l, s}\)). For $j$ components, \(\widetilde{\mathcal{B}}_{t, l, s}\) is the juxtaposition of \(\mathcal{B}_1, \mathcal{B}_2, \ldots, \mathcal{B}_J\), where \(\mathcal{B}_1\) provides the first \(k_1\) columns, \(\mathcal{B}_2\) provides the next \(k_2\) columns, and so on, with \(\mathcal{B}_J\) providing the final \(k_J\) columns. The \((t, l)\)-equidistribution of this combined generator is stated in Proposition \ref{prop3} \cite{ME3}.

\begin{proposition} [\cite{ME3}]\label{prop3}
The sequence is \((t, l)\)-equidistributed if and only if the matrix \(\widetilde{\mathcal{B}}_{t, l, s}\) has full rank \(tl\). If \(lt = \lfloor k / t \rfloor < k / t \leq L\), then the sequence is also \(\text{CF}(t)\) if and only if the matrix \(\widetilde{\mathcal{B}}_{t, lt + 1, s}\) has rank \(k\).
\end{proposition}

\subsubsection{Maximal Equidistribution}
\label{section2.5}

Suppose a unit hypercube with \(t\) dimensions is divided into \(2^{t\ell}\) equally sized cubic cells. A series of \(t\)-dimensional points, or vectors, are generated by the generator and placed into the cells based on their respective values. The parameter \(\ell_t\) represents the resolution in dimension \(t\), defined as the largest \(\ell\) for which each cell contains the same number of points. The generator is referred to as maximally equidistributed (ME) when it reaches the maximum value of \(\ell_t\) across all dimensions \cite{ME3}.

To determine whether the sequence is maximally equidistributed, it is not necessary to compute $l_t$ for $ t = 1, 2, \ldots, k$. The following Proposition \ref{prop4} states that a maximal period sequence is maximally equidistributed if and only if $\Lambda_t$ = 0 for all $t \in \Phi_1 \cup \Phi_2$ \cite{ME3}. 

\begin{proposition} [\cite{ME3}]\label{prop4}
A maximal period sequence is ME if and only if \(\Lambda_t = 0\) for all \(t \in \Phi_1 \cup \Phi_2\). It is also ME if and only if \(\Delta_l = 0\) for all \(l \in \Psi_1 \cup \Psi_2\). Here $L$ is the word size of a computer,
\begin{align}\label{eq:phi}
\scriptsize
\Phi_1 &= \left\{ \max\left( 2, \left\lfloor \frac{k}{L} \right\rfloor \right), \dots, \left\lfloor \sqrt{k} \right\rfloor \right\} \\
\Phi_2 &= \left\{ t = \left\lfloor \frac{k}{\ell} \right\rfloor \,\middle|\, \ell = 1, \dots, \left\lfloor \sqrt{k} \right\rfloor \right\}\\
\Psi_{1} &= \{1, \ldots, \lfloor \sqrt{k} \rfloor \} \\
\Psi_{2} &= \{\ell = \lfloor k/t \rfloor \mid 
t = \max(2, \lfloor k/L \rfloor), \ldots, \lfloor \sqrt{k-1} \rfloor \}
\end{align}
\end{proposition}
Here, the resolution in dimension $t$ ($l_t$) denotes the highest value of $l\leq L$ for which the sequence is $(t,l)$-equidistributed, where $l_t\leq l_{t}^{*}$ and $l_{t}^{*}$=$\min(L, \lfloor k/t \rfloor)$.

\section{Maximal Length CAs and Linear CA-based PRNGs}
\label{section2.6}
A Cellular Automaton (CA) is a mathematical model that contains a grid of cells where each cell may take any of the states from a finite state set $\mathcal{S}$. Each cell updates its state using the current cell's state and states of its \emph{neighbors} using a local rule (${R}$). These neighborhood combinations are called \emph{Rule Min Term} (RMT).  A collection of states at each time step $t$ is called a \emph{configuration}. All cells are updated in parallel by the local rule(s); so the CA goes from one configuration to the next configuration. 


In this work, we consider finite Elementary Cellular Automaton (ECA), which is a one-dimensional, 2-state 3-neighborhood CA -- each cell depends on the states of its left neighbor, itself, and its right neighbor to update its state. The CA has $n$ cells under null boundary condition; that is, the cells are numbered from $0$ to $n-1$ and the neighbors of the terminal cells are set to zero ($0$). So, there are $2^3 =8$ RMTs possible with the neighborhood combinations from 000 to 111. There are totally $2^{2^3}=256$ ECA rules, numbered from $0$ to $255$. For instance, ECA rules 90 and 150 are defined as:
\[
\begin{aligned}
\text{Rule 90:} \quad & S_i^{t+1} = S_{i-1}^{t} \oplus S_{i+1}^t, \\[6pt]
\text{Rule 150:} \quad & S_i^{t+1} = S_{i-1}^t \oplus S_i^t \oplus S_{i+1}^t.
\end{aligned}
\]
where $S_i^t$ is the state of cell $i$ at the $t^{th}$ time step.



In a CA, if each cell updates its state by using the same rule $R$, it is called a \emph{uniform} CA. Otherwise, it is called a \emph{non-uniform} CA. In an $n$-cell non-uniform CA, instead of a single rule $R$, a rule vector \(\mathcal{R} = \langle R_0, R_1, \ldots, R_{n-1} \rangle\) is used, where the $i^{th}$ uses rule $R_i$ to update its state. 
In a CA, if the transition function (local rule) is a linear map, then the CA is said to be a \emph{linear} CA. Here, the configuration space may be considered as a vector space.
For example, ECA rules 90 and 150 are linear CAs. An notable type of linear non-uniform CAs are called the \emph{maximal length} CAs. An important characteristic of this CA is that it provides the maximal period length of \(2^n-1\). It has been shown that, a linear CA can be a maximal length CA if and only if the rule vector includes only the rules 90 and 150, and the characteristic polynomial is primitive over GF(2) \cite{CA7}. Additionally, only a certain combination of rules 90 and 150 over the null boundary condition produces the maximal length CAs.
Linear maximal length CAs can also be characterized by linear algebra. Any linear maximal length CA can be represented by $n \times n$ \emph{characteristic matrix}. 
\begin{definition}\label{def:tmatrix}
	The characteristic matrix \textit{T} is a matrix of order $n$ $\times$ $n$ that is defined as follows: 
	\[
	T[i,j] = 
	\begin{cases}
		1,& \text{if } \text{cell i depends on cell j}\\
		0,              & \text{otherwise}
	\end{cases}
	\]
\end{definition}

For instance, the characteristics matrix of a 5-cell maximal length CA with \(\mathcal{R}\)= $ \langle150, 90, 90, 90, 90 \rangle $ is shown below. The characteristic polynomial of this matrix is primitive --  the CA has a cycle (period) of length $2^5-1=31$. 
\[
T =
\begin{bmatrix}
1 & 1 & 0 & 0 &0 \\
1 & 0 & 1 & 0 &0\\
0 & 1 & 0 & 1 &0\\
0 & 0 & 1 & 0 &1\\
0 & 0 & 0 & 1 &0\\
\end{bmatrix}
\]

Several one-dimensional linear maximal length CAs of size \(n = 1 \text{ to } 500\) are represented in \cite{CA4}, which uses rule 150 for a maximum of 2 cell positions, and all remaining positions use the same rule 90. In Ref \cite{CA7} two almost uniform maximal length CAs, named as $CA (90')$ and $CA (150')$ of different sizes are introduced, which also provides a maximal length of $2^{n}-1$ when the size $n$ is a Sophie Germain prime \cite{CA7}. The $CA (90')$ indicates that all the cells use rule 90, except the first cell, which uses rule 150. Similarly, $CA (150')$ means that all the cells use rule 150 except the first, which uses rule 90. So, except the first cell, all the remaining positions use the same rule.

\subsection{Linear CA-based PRNGs}
\label{section3}


Since linear CAs can be efficiently implemented using linear algebraic operations and linear maximal length CAs have similar properties like LFSRs, they are utilized in the design of PRNGs \cite{maxCA33,parallel24}. The linear maximal length CA based PRNGs were introduced as an efficient model for hardware implementation \cite{maxCA33}. In this paper, the 32-bit maximal length CA \(\mathcal R_1 = \langle 90,150,90,90,90,150,150,90,90,90,90,90,150,90,90,150,150,90,150,150,\\150,90,150,150,150,150,90,150,90,150,90,150 \rangle\) is used for generating pseudo-random numbers for Built-In Self-Test (BIST). In \cite{linear16}, linear $CA(150')$ with size $n$=1409 (named here as $\mathcal{R}_2$) is used for designing a PRNG that has speed and randomness test results comparable to the Mersenne twister. In \cite{parallel24}, a multiple-stream parallel PRNG is designed utilizing 35-bit $CA(150')$ (named here as $\mathcal{R}_3$); it generates a 32-bit random number with only 3 bits of waste, but it failed in most of the statistical testbeds. Then, considering the current computer word size, a 64-bit multiple stream parallel PRNG is developed using a 64-bit maximal length which uses rule 150 at $3^{rd}$ and $5^{th}$ cells and rule 90 at other cells (named here as $\mathcal{R}_4$) \cite{parallel25}; it also passes nearly all the tests in testbeds after utilizing tempering. However, although the linear CA-based PRNGs provide a good period length and speed; the theoretical figures of merit (equidistribution) of these linear CA-based PRNGs has never been examined before. So, we test them in the next section.

\subsection{Equidistribution of Exisiting Linear CA-based PRNGs}
Here, we first take these linear maximal length CAs of sizes 32 ($\mathcal{R}_1$), 35 ($\mathcal{R}_3$), 64 ($\mathcal{R}_4$), and 1409 ($\mathcal{R}_2$) and analyze the theoretical properties (equidistribution) using the steps discussed in Section~\ref{section2.4}. To analyze the equidistribution of this type of PRNGs, we first fit the CA-based PRNG into the framework described in Section \ref{section2.2}. Suppose the rule vector is of size $k$. Then, matrix $A$ is the transition matrix $T$ of size $k \times k$ for the rule vector following Definition~\ref{def:tmatrix} and matrix $B$ is $w$ rows of the $k \times k$ identity matrix.

For example, for a rule vector \(\mathcal R_1 \) of size 32, we test the maximal equidistribution based on Proposition~\ref{prop3}. According to this proposition, if for all $t$-values $(t,l)$ equidistribution is satisfied, then only we say that it is maximally equidistributed. Here $t \in \{2, 3, 4, 5, 6, 8, 10, 16, 32\}$. Among these $t$ values, for only $t=32$, it achieves the $ (t,l)$ equidistribution. Therefore, this linear CA-based PRNG with size 32 is not maximally equidistributed. Following \autoref{table1} gives the result of this equidistribution for the rule vector \(\mathcal R1 \). Similarly, the remaining CAs $\mathcal{R}_2, \mathcal{R}_3$ and $\mathcal{R}_4$ also fail to satisfy the maximal equidistribution.

\begin{scriptsize}
\begin{longtable}{|>{\centering\arraybackslash}p{1cm}|
                        >{\centering\arraybackslash}p{2cm}|
                        >{\centering\arraybackslash}p{2cm}|
                        >{\centering\arraybackslash}p{2cm}|
                        >{\centering\arraybackslash}p{4cm}|}
\caption{Equidistribution Test of linear CA-based PRNG for the Rule vector \(\mathcal R1 \) 
\label{table1}}\\
\hline
\textbf{t} & \textbf{$l_t^*=\min(L, \lfloor k/t \rfloor)$} & \textbf{$l_t \ (l_t \le l_t^*)$} & \textbf{Rank} & \textbf{Equidistribution} \\
\hline
\endfirsthead

\hline
\textbf{t} & \textbf{$l_t^*=\min(L, \lfloor k/t \rfloor)$} & \textbf{$l_t \ (l_t \le l_t^*)$} & \textbf{Rank} & \textbf{Equidistribution} \\
\hline
\endhead

\hline \multicolumn{5}{r}{{}} \\
\endfoot

\hline
\endlastfoot

2  & 16 & 16 & 18 & not $(t,l)$-equi-distributed \\
3  & 10 & 10 & 13 & not $(t,l)$-equi-distributed \\
4  & 8 & 8 & 12 & not $(t,l)$-equi-distributed \\
5  & 6 & 6 & 11 & not $(t,l)$-equi-distributed \\
6  & 5 & 5 & 11 & not $(t,l)$-equi-distributed \\
8  & 4  & 4  & 12 & not $(t,l)$-equi-distributed \\
10 & 3  & 3  & 13 & not $(t,l)$-equi-distributed \\
16 & 2  & 2  & 17 & not $(t,l)$-equi-distributed \\
32 & 1  & 1  & 32 & $(t,l)$-equi-distributed \\
\end{longtable}
\vspace{-1em}
\end{scriptsize}
As a result, all of these existing linear CA-based PRNGs do not achieve maximal equidistribution. To address this issue, in the next sections, our target is to design CA-based PRNGs which are light-weight but still satisfy the theoretical quality criteria essential for a good PRNG. 


\section{Source of Component CAs in Combined CA-based PRNG}
\label{section4}
The primary goal of our work is to develop CA-based PRNGs that achieve large period and maximal equidistribution. As discussed in Section~\ref{section1}, there are two established ways for LFSR based generators to achieve this requirement -- (1) use of proper output transformation function or \emph{tempering} or (2) design a combined generator with proper parameters. Now, tempering is an external function and will require extra computation units (hardware). Also, as in the case for CAs $\mathcal{R}_2$ or $\mathcal{R}_4$, tempering may not provide equidistribution if the underlying primitive polynomial is not of very large degree.

Since we want our design to be hardware efficient, light-weight and completely based on the properties of CA, we prefer not to use any output transformation or \emph{tempering}. So, our choice is to design combined generators with linear maximal length CAs that achieve the maximal equidistribution. To make sure the PRNGs are light-weight and resource efficient, we restrict our component CAs' number of cells $n$ to be up to $128$ so that only $128$ bits are sufficient. To obtain a long period, we select the components of the combined CA-based PRNG that have a maximal period of $(2^n-1)$ and periods that are pairwise relatively prime. 

As already mentioned, a CA is easily implementable on hardware. A cell can be represented by one bit of memory or a flip-flop (FF) with a combinational circuit for the rule that updates its state. For example, hardware implementation of the maximal length CA \(\mathcal{R}\)= $ \langle150, 90, 90, 90, 90 \rangle $ is represented in \autoref{figure1}.
\begin{figure}[!h]  
    \centering
    \fbox{%
        \includegraphics[width=0.5\textwidth]{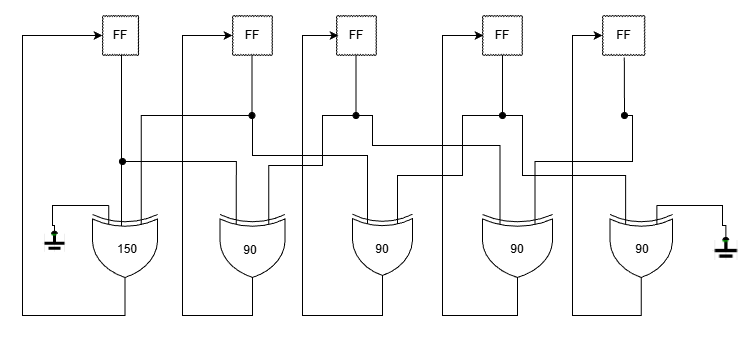}%
    }
\caption{Hardware implementation of a maximal length CA $ \langle150, 90, 90, 90, 90 \rangle $}
\label{figure1}
\vspace{-1em}
\end{figure}
Here, we can see that, rule 90 needs only one 2-input xor gate and 150 needs one 3-input xor gate. So, in the chosen CA, if almost all the cells have the same rule, the same hardware module for a cell can be duplicated making the design more cost-effective. 

Now, as we want to make our generator resource-efficient, we look for such maximal length CAs where majority of the cells use the same rule. In literature, there are two resources available that provide such a list of maximal length CAs given in Ref. \cite{CA4} and Ref. \cite{CA7}. In Ref. \cite{CA4} a list of linear maximal length CAs up to degree 500 is given where only two cells use rule 150 and all other cells use the same rule 90. Whereas, Ref. \cite{CA7} gives two imperfect strategies to generate primitive polynomials of degree $n$ using CA (90') and CA (150'). It has been claimed that, these CAs where only the first cell uses rule 150 (resp. rule 90) and all other cells use rule 90 (resp. 150) are able to generate a period of $2^{n}-1$ when $n$ is a Sophie-Germain prime.  These kind of maximal length CAs are highly efficient for hardware implementation, because most of the cells follow the same rules, with only one or two exceptions in specific positions. 


However, the most important question is whether these CAs have the potential to be a good source of randomness for designing a good PRNG. It has been shown that, a CA can be a good source of randomness only when it is \emph{chaotic} (when defined over infinite lattice) and satisfies some desirable properties like no bias towards any particular state, there is \emph{flow of information} on both directions and very large period \cite{KamalikaPhDThesis}. Since, we are taking maximal length CAs, automatically they satisfy the large period criteria. Moreover, the rules $90$ and $150$ are \emph{balanced} rules since the update function is not biased to generate any state. So, next, we need to see if the CAs have flow of information and are chaotic. In Ref \cite{chaotic23}, a method has been designed to calculate the flow of information in both left and right directions and predict the chaotic behavior of non-uniform CAs using a parameter, called the $P$-parameter. 

To predict the chaotic property by $P$-parameter, information flow is calculated in two parts. First, the probability of how much the neighbors of a cell are affected by a minor change in that cell is calculated; this process is called \emph{information propagation}. Next, it assesses the effect on the neighbors, which may in turn affect the original cell. It is called \emph{information cooking}. The $p$-parameter is an ordered pair $(p_1, p_2)$. In abstract terms, if there is always an information flow over the grid, then, $p_1$ is the minimum of maximum values of information propagation in both directions for each cell. Similarly, for calculating $p_2$, for each cell, first find the minimum of the maximum values between information propagation and information cooking in each of the directions. The minimum of these values over all cells is $p_2$. If both $p_1,p_2$ are high, the CA is considered as chaotic and $(0.75, 0.5)$ are considered as threshold values for chaos. owever, since our CAs are finite to be used as PRNGs, we cannot directly apply this parameter. 

To address this and get an idea of the behavior when defined over infinite lattice, we find the information propagation and information cooking considering any cell having rule 90 or rule 150 replicating over the grid and calculate $P$-parameter value. For rule 90, there is always information propagation as well as information cooking in each of the directions, so probability is 100\%. However, for rule 150, there is always information propagation in each direction, but no information cooking. Therefore, whenever a rule vector uses both rule 90 and 150, the value of $P$-parameter is $(p_1,p_2)=(1,1)$. For example, take the 5-cell linear maximal length CA with \(\mathcal{R}\)= $ \langle 150, 90, 90, 90, 90\rangle $. For this CA, the parameter value is $(p_1,p_2)=(1,1)$; therefore, it satisfies the chaotic property. Since all our linear maximal length CAs use both rule 90 and rule 150, all of them satisfy the chaotic property. So, due to the characteristic of efficient hardware implementation and chaos, these maximal length CAs are used to design light-weight combined CA-based PRNGs, which are discussed in the following sections.

\section{Combined Maximal Length CA-based PRNGs: Phase I}
\label{section4.1}
A combined PRNG usually consists of more than one component. Here, each maximal length CA is considered as a component. Initially, two maximal length CAs are taken  whose periods are relatively prime. Both CAs are evolved for $t$ time steps. At each and every time step, both CA configurations are combined by using the XOR operation (using either left or right padding), and we get a new random sequence. This sequence is called a combined random number sequence. This process is represented in Algorithm~\ref{alg:1}.

\begin{algorithm}[!h]
\caption{Combined two-component maximal length CA-based PRNG}\label{alg:1}
\scriptsize{
\begin{algorithmic}[1]
\REQUIRE Two maximal length CAs $CA_1,CA_2$ with relatively prime periods
\REQUIRE Number of time steps $t$ and Size of each CA rule vector $k_j$
\REQUIRE Seed (configurations)  $S_1^{(0)}\in\{0,1\}^{k_1}$, $S_2^{(0)}\in\{0,1\}^{k_2}$
\ENSURE $U$ \COMMENT{Combined random sequence}
\STATE $U \gets [\;]$
\STATE Initialize each maximal length  CA $CA_j$ with a random configuration (seed)$S_j$ of length $k_j$

\FOR{$n = 1$ \TO $t$}
    \STATE $u_1^{(n)} \gets \textsc{NextState}(CA_1)$
    \STATE $u_2^{(n)} \gets \textsc{NextState}(CA_2)$
    \STATE $u_n \gets u_1^{(n)} \oplus u_2^{(n)}$ 
    \STATE Append $u_n$ to $U$
\ENDFOR

\RETURN $U$
\end{algorithmic}}
\end{algorithm}

Let's consider $k$ to be the size of a maximal length CA with a period $\rho = 2^k - 1$, and let $u_n^j$ denote a random sequence generated by the $j^{th}$ maximal length CA. The combined random sequence is represented by: 
\[
u = u_n^1 \oplus u_n^2 \oplus \dots \oplus u_n^j
\]

In this work, we focus on two components, so $j$=2, so $u_n = u_n^1 \oplus u_n^2$, where $u_n^1$ is the sequence generated from the first CA and $u_n^2$ is the sequence generated from the second CA. Additionally, each component in the combined CA-based PRNG has the period of $\rho_j=2^{k_j}-1$ and period length of two components are relatively prime, so the period of this two component combined CA-based PRNG is equal to $\mathrm{LCM}(2^{k_1} - 1, 2^{k_2} - 1) = (2^{k_1} - 1)(2^{k_2} - 1)$.  Then, we test equidistribution of these combined PRNGs.

\begin{center}
\scriptsize
\setlength{\tabcolsep}{6pt}
\renewcommand{\arraystretch}{0.9}
\vspace{-1.5em}
\begin{longtable}{|c|c|c|c|c||c|c|c|c|c|}
\caption{Characteristics of Component Maximal Length CAs}
\label{table4mca}\\
\hline
\textbf{k} & \textbf{Maximal Length CA} & \textbf{k/2} & $\mathbf{N_1}$ & $\mathbf{N_1/k}$ & \textbf{k} & \textbf{Maximal Length CA} & \textbf{k/2} & $\mathbf{N_1}$ & $\mathbf{N_1/k}$\\  \hline
\endfirsthead

\hline
\textbf{k} & \textbf{Maximal Length CA} & \textbf{k/2} & $\mathbf{N_1}$ & $\mathbf{N_1/k}$ & \textbf{k} & \textbf{Maximal Length CA} & \textbf{k/2} & $\mathbf{N_1}$ & $\mathbf{N_1/k}$\\ 
\hline
\endhead

29 & 1 & 14 & 11 & 0.38 & 79 & 9 & 39 & 11 & 0.14 \\ \hline
30 & 1 & 15 & 9 & 0.30 & 80 & 1,71 & 40 & 31 & 0.39 \\ \hline
31 & 11 & 15 & 5 & 0.16 & 81 & 1 & 40 & 25 & 0.30 \\ \hline
32 & 1,15 & 16 & 11 & 0.34 & 82 & 1,69 & 41 & 33 & 0.40 \\ \hline
33 & 1 & 16 & 11 & 0.33 & 83 & 1 & 41 & 27 & 0.32 \\ \hline
34 & 1,19 & 17 & 11 & 0.32 & 84 & 36 & 42 & 39 & 0.46 \\ \hline
35 & 1 & 17 & 13 & 0.37 & 85 & 1,46 & 42 & 27 & 0.32 \\ \hline
36 & 6 & 18 & 16 & 0.44 & 86 & 1 & 43 & 31 & 0.36 \\ \hline
37 & 9 & 18 & 17 & 0.46 & 87 & 13 & 43 & 15 & 0.17 \\ \hline
38 & 7 & 19 & 13 & 0.34 & 88 & 5 & 44 & 41 & 0.47 \\ \hline
39 & 1 & 19 & 13 & 0.33 & 89 & 1 & 44 & 29 & 0.33 \\ \hline
40 & 8 & 20 & 17 & 0.42 & 90 & 1 & 45 & 31 & 0.34 \\ \hline
41 & 1 & 20 & 17 & 0.41 & 91 & 15 & 45 & 19 & 0.21 \\ \hline
42 & 19 & 21 & 23 & 0.55 & 92 & 3,71 & 46 & 45 & 0.49 \\ \hline
43 & 3 & 21 & 17 & 0.39 & 93 & 33 & 46 & 21 & 0.22 \\ \hline
44 & 4,26 & 22 & 21 & 0.47 & 94 & 42 & 47 & 25 & 0.26 \\ \hline
45 & 9 & 22 & 17 & 0.37 & 95 & 1 & 47 & 13 & 0.14 \\ \hline
46 & 2,10 & 23 & 19 & 0.41 & 96 & 6 & 48 & 35 & 0.36 \\ \hline
47 & 13 & 23 & 9 & 0.19 & 97 & 1,82 & 48 & 27 & 0.28 \\ \hline
48 & 15 & 24 & 19 & 0.40 & 98 & 8 & 49 & 35 & 0.36 \\ \hline
49 & 1,10 & 24 & 21 & 0.43 & 99 & 13 & 49 & 25 & 0.25 \\ \hline
50 & 11 & 25 & 27 & 0.54 & 100 & 1,67 & 50 & 27 & 0.27 \\ \hline
51 & 1 & 25 & 17 & 0.33 & 101 & 1,20 & 50 & 33 & 0.32 \\ \hline
52 & 2,29 & 26 & 17 & 0.32 & 102 & 33 & 51 & 35 & 0.34 \\ \hline
53 & 1 & 26 & 21 & 0.40 & 103 & 15 & 51 & 21 & 0.20 \\ \hline
54 & 9 & 27 & 25 & 0.46 & 104 & 2,40 & 52 & 27 & 0.26 \\ \hline
55 & 17 & 27 & 11 & 0.20 & 105 & 1 & 52 & 31 & 0.29 \\ \hline
56 & 4,14 & 28 & 25 & 0.45 & 106 & 30 & 53 & 45 & 0.42 \\ \hline
57 & 9 & 28 & 17 & 0.29 & 107 & 19 & 53 & 25 & 0.23 \\ \hline
58 & 17 & 29 & 25 & 0.43 & 108 & 1,35 & 54 & 37 & 0.34 \\ \hline
59 & 4,15 & 29 & 17 & 0.28 & 109 & 1,4 & 54 & 41 & 0.38 \\ \hline
60 & 2,38 & 30 & 21 & 0.35 & 110 & 13 & 55 & 37 & 0.34 \\ \hline
61 & 1,10 & 30 & 27 & 0.44 & 111 & 27 & 55 & 13 & 0.12 \\ \hline
62 & 5 & 31 & 19 & 0.30 & 112 & 2,5 & 56 & 43 & 0.38 \\ \hline
63 & 31 & 31 & 3 & 0.04 & 113 & 1 & 56 & 23 & 0.20 \\ \hline
64 & 3,5 & 32 & 21 & 0.32 & 114 & 22 & 57 & 43 & 0.37 \\ \hline
65 & 1 & 32 & 13 & 0.20 & 115 & 41 & 57 & 19 & 0.16 \\ \hline
66 & 1,19 & 33 & 21 & 0.31 & 116 & 16 & 58 & 37 & 0.32 \\ \hline
67 & 15 & 33 & 17 & 0.25 & 117 & 33 & 58 & 27 & 0.23 \\ \hline
68 & 8 & 34 & 23 & 0.34 & 118 & 30 & 59 & 37 & 0.31 \\ \hline
69 & 1 & 34 & 23 & 0.33 & 119 & 1 & 59 & 19 & 0.16 \\ \hline
70 & 1,37 & 35 & 17 & 0.24 & 120 & 3,73 & 60 & 37 & 0.30 \\ \hline
71 & 17 & 35 & 15 & 0.21 & 121 & 45 & 60 & 19 & 0.15 \\ \hline
72 & 6,55 & 36 & 29 & 0.40 & 122 & 14 & 61 & 35 & 0.29 \\ \hline
73 & 9 & 36 & 25 & 0.34 & 123 & 51 & 61 & 15 & 0.12 \\ \hline
74 & 1 & 37 & 29 & 0.39 & 124 & 21 & 62 & 33 & 0.27 \\ \hline
75 & 7 & 37 & 19 & 0.25 & 125 & 13 & 62 & 19 & 0.15 \\ \hline
76 & 2,22 & 38 & 33 & 0.43 & 126 & 40 & 63 & 21 & 0.16 \\ \hline
77 & 3,44 & 38 & 29 & 0.38 & 127 & 15 & 63 & 11 & 0.08 \\ \hline
78 & 1,41 & 39 & 19 & 0.24 & 128 & 1,29 & 64 & 27 & 0.21 \\ \hline

\end{longtable}\vspace{-1.5em}
\end{center}

\subsection{Considering at maximum 2 cells with Rule 150 \cite{CA4}}
We first take the maximal length CAs of Ref. \cite{CA4} with degrees between $32$ and $64$ where at maximum two cells have rule 150 and all other cells have rule 90. \autoref{table4mca} lists these CAs. Here, $k$ indicates the size or the degree of the polynomial of the corresponding linear maximal length CA. The corresponding rule vector is represented by the positions of cells using rule 150. For example, a maximal length CA of size 32 contains 1 and 15; it means that the cell numbered 2 and 15 only use rule 150, and the remaining cells use rule 90. In this way, we initially take all the two-component combinations from this table whose period lengths are relatively prime.  There are 306 possible combinations in total. We find that, even though every combination achieves a large period close to $2^{k_1 + k_2 }$, none of them attain the maximal equidistribution.

Since none of two-component combined CA-based PRNG  satisfy the maximally equidistributed property, we also increase the number of components from two to three to see if there is any improvement. We select the 3-component combinations $(k_1, k_2,k_3)$ within the range of $32 \leq k_1, k_2,k_3 \leq 64$. A total of 1433 pairwise relatively prime combinations have been found. These combinations also have a period length close to $2^{k_1+k_2+k_3}$, but we again observe that maximal equidistribution is not satisfied.  Therefore, these two-component and three-component combined CA-based PRNGs attain a large period that is close to $2^{k_1 + k_2 + \dots + k_j}$, but fail to satisfy the maximally equidistribution property necessary as theoretical figures of merit for randomness. So, we do not increase the number of components any further. 

\begin{example}\label{ex:1}
 Consider two maximal length CAs of sizes 31$(k_1)$ and 32$(k_2)$  respectively; the periods are $\rho_1 = 2^{31} - 1$ and $\rho_2 = 2^{32} - 1$. The corresponding rule vectors are \(\mathcal R1\)=$ \langle90, 90, 90, 90, 90, 90, 90, 90, 90, 90, 150, 90, 90, 90, 90, 90,\\ 90, 90, 90, 90, 90, 90, 90, 90, 90, 90, 90, 90, 90, 90, 90 \rangle$ and \(\mathcal R2\)=$ \langle 150, 90, 90, 90, 90,\\ 90, 90, 90, 90, 90,90,90, 90, 90, 150, 90, 90, 90, 90, 90, 90, 90, 90, 90, 90, 90, 90, 90,\\ 90, 90, 90, 90 \rangle$. The periods of these two CAs are relatively prime since $\gcd(2^{31} - 1,\, 2^{32} - 1) = 1$; therefore, combining both CAs configuration achieves a large period of approximately $2^{63}$. Then to test the equidistribution let us first compute the values of $\phi_1$ and $\phi_2$ using Equation~\ref{eq:phi}; here $\Phi_1 = \{2,3,4,5,6,7\}$ and $\Phi_2= \{9,10,12,15,21,31,63\}$. Therefore, $t= \Phi_1 \cup \Phi_2 = \{2,3,4,5,6,7,9,10,12,\\15,21,31,63\}$. However, out of all the $t$ values, only two $t$ values satisfy the $(t,l)$-equidistribution. So, this PRNG does not attain maximal equidistribution. These results are displayed in the following \autoref{table2}.
\end{example}

\begin{scriptsize}
\vspace{-1.5em}
\begin{longtable}{|>{\centering\arraybackslash}p{1cm}|
                        >{\centering\arraybackslash}p{3cm}|
                        >{\centering\arraybackslash}p{2cm}|
                        >{\centering\arraybackslash}p{2cm}|
                        >{\centering\arraybackslash}p{4cm}|}
\caption{Equidistribution of combined CA-PRNG (31,32)\label{table2}}\\
\hline
\textbf{t} & \textbf{$l_t^*=\min(L, \lfloor k/t \rfloor)$} & \textbf{$l_t \ (l_t \le l_t^*)$} & \textbf{Rank} & \textbf{Equidistribution} \\
\hline
\endfirsthead
\multicolumn{5}{c}{{\tablename\ \thetable{} -- continued from previous page}} \\
\hline
\textbf{t} & \textbf{$l_t^*=\min(L, \lfloor k/t \rfloor)$} & \textbf{$l_t \ (l_t \le l_t^*)$} & \textbf{Rank} & \textbf{Equidistribution} \\
\hline
\endhead
\hline 
\endfoot
\hline
\endlastfoot
2  & 31 & 31 & 39 & not $(t,l)$-equi-distributed \\
3  & 21 & 21 & 33 & not $(t,l)$-equi-distributed \\
4  & 15 & 15 & 29 & not $(t,l)$-equi-distributed \\
5  & 12 & 12 & 29 & not $(t,l)$-equi-distributed \\
6  & 10 & 10 & 26 & not $(t,l)$-equi-distributed \\
7  & 9  & 9  & 27 & not $(t,l)$-equi-distributed \\
9  & 7  & 7  & 29 & not $(t,l)$-equi-distributed \\
10 & 6  & 6  & 29 & not $(t,l)$-equi-distributed \\
12 & 5  & 5  & 31 & not $(t,l)$-equi-distributed \\
15 & 4  & 4  & 34 & not $(t,l)$-equi-distributed \\
21 & 3  & 3  & 45 & not $(t,l)$-equi-distributed \\
31 & 2  & 2  & 62 & $(t,l)$-equi-distributed \\
63 & 1  & 1  & 63 & $(t,l)$-equi-distributed \\

\end{longtable}\vspace{-1.5em}
\end{scriptsize}

\subsection{Considering $CA(90')$ or $CA(150')$}
Next, we increase the size of the CAs up to $128$ while taking $CA(90')$ and $CA(150')$ from Ref. \cite{CA7}. Such CAs of sizes close to 32, 64 and 128 are combined and analyzed for period length and equidistribution property. The maximal length $CA(90')$ and $CA(150')$ close to $128$ are of cell lengths $105, 113, 119$, close to 64 are of cell lengths 65, 69 and close to 32 are 26, 29, 35, 39. Here, we combine these CAs in the following three ways: $(i)$ $CA(90')$ with $CA(90')$ $(ii).$ $CA(150')$ with $CA(150')$ (iii). $CA(90')$ with $CA(150')$ and  $CA(150')$ with $CA(90').$

The two-component relatively prime combinations for the CA sizes close to 32, 64 and 128 for combined generator are displayed in \autoref{table4}). Here also, we observe that, the PRNGs achieve the period length close to $2^{k_1+k_2}$ where $k_1$ is the size of the first component and $k_2$ is the size of the second component. However, when tested for the maximal equidistribution property, here also, each of them fails to satisfy the maximal equidistribution. Therefore the combined PRNGs using $CA(90')$ and $CA(150')$ also attain the period close to the maximal, but does not achieve the maximally equidistributed characteristic. In the next section, we further improve the PRNGs to achieve this maximal equidistribution.

\begin{scriptsize}
\begin{longtable}{|>{\centering\arraybackslash}p{2cm}|%
                >{\centering\arraybackslash}p{3.4cm}|%
                >{\centering\arraybackslash}p{3.4cm}|%
                >{\centering\arraybackslash}p{3.4cm}|}
\caption{Combined PRNG combinations using $CA(90')$ and $CA(150')$}
\label{table4}\\
\hline
\multirow{2}{*}{\textbf{CA}} & 
\multicolumn{3}{c|}{\textbf{Combined PRNGs with Relatively Prime Combination}} \\ \cline{2-4}
& \textbf{$CA(90')$ with $CA(90')$ \newline $(k_1,k_2)$ : $(90',90')$} &
\textbf{$CA(150')$ with $CA(150')$ \newline $(k_1,k_2)$ : $(150',150')$} &
\textbf{$CA(90')$ with $CA(150')$ \newline $(k_1,k_2)$ : $(90',150')$} \\ \hline

Close to 32 &
(26,29), (26,35), (29,35), (29,39), (35,39) &
(26,29), (26,35), (29,35), (29,39), (35,39) &
(26,29), (29,26), (26,35), (35,26), (29,35), (35,29), (29,39), (39,29), (35,39), (39,35) \\ \hline

Close to 64 &
(65,69) &
(65,69) &
(65,69), (69,65) \\ \hline
Close to 128 &
(105,113), (113,119) &
(105,113), (113,119) &
(105,113), (113,105), (113,119), (119,113) \\ \hline
\end{longtable}
\vspace{-1.5em}
\end{scriptsize}

\section{Combined Maximal Length CA-based PRNGs: Phase II}
\label{section4.2}
One of the problems of chaotic CAs is they have tendency to generate self-similar and self-organized patterns like Sierpinski triangles and equilateral triangles \cite{KamalikaPhDThesis}. This makes the generator non-random. So, to achieve better randomness quality by breaking these symmetrical patterns, Ref. \cite{maxCA33} introduced the concept of site spacing.
This means instead of using the configuration of each time step, some steps are skipped. This can improve the random numbers quality and can give better equidistribution.

To develop such a PRNG, we take two maximal length CAs with period lengths that are relatively prime, and run for $t$ time steps. Now, instead of combining the configuration of the CAs at each time step, skip $s-1$ steps and combine the configuration at each $s^{th}$ step. For instance, if $s$ = 2, combine the sequence after each second step. This process is shown in Algorithm~\ref{algo:2}. 

\begin{algorithm}[H]
\caption{Combined two-component maximal length CA-based PRNG with time spacing}\label{algo:2}
\scriptsize
\begin{algorithmic}[1]
\REQUIRE Two maximal length CAs $CA_1, CA_2$ with relatively prime periods $\rho_1,\rho_2$
\REQUIRE Number of time steps $t$, step size $s$,  $s$ must satisfy $\gcd\!\left(s,\,\rho1 \rho_2\right) = 1$, and sizes of each CA rule vector $k_1, k_2$
\REQUIRE Initial configurations (seeds) $S_1^{(0)} \in \{0,1\}^{k_1}$, $S_2^{(0)} \in \{0,1\}^{k_2}$
\ENSURE $U$ \COMMENT{Combined random sequence}

\STATE $U \gets [\;]$
\STATE Initialize each CA $CA_j$ with its seed $S_j$

\FOR{$n = 1$ \TO $t$}
    \FOR{$i = 1$ \TO $s$}
        \STATE $u_1^{(n)} \gets \textsc{NextState}(CA_1)$
        \STATE $u_2^{(n)} \gets \textsc{NextState}(CA_2)$
    \ENDFOR
    \STATE $u_n \gets u_1^{(n)} \oplus u_2^{(n)}$ 
    \STATE Append $u_n$ to $U$
\ENDFOR

\RETURN $U$
\end{algorithmic}
\end{algorithm}

Since, configurations of each $s-1$ steps are skipped, in terms of the PRNG, each component can be represented by the matrix $A$ = $(T)^s$, where $T$ is the characteristic matrix for the component maximal length CA.
Consider $\rho_i$ and and $\rho_j$ be the period lengths of the first and second component. Then $\rho$ is the LCM of two component periods; that is, $\rho = \operatorname{lcm}(\rho_i, \rho_j)$. If $\gcd(s, \rho) = 1$, then this combined PRNG with time spacing achieves a large period close to $2^{k_1 + k_2}$. 
Otherwise, the period is $\frac{\rho}{\gcd(s, \rho)}$. These can be derived from the following proposition, similar to Proposition~\ref{prop1}. 

\begin{proposition}\label{prop5}
Let $k_1$ and $k_2$ be the 
degrees of two primitive polynomials (modulo two) 
with periods $\rho_1$ and $\rho_2$ such that $\gcd(\rho_1, \rho_2) = 1$ and $s>1$ is the time spacing size. If $\gcd(s, \rho) = 1$, then the period of the combined generator with $s$ time spacing is close to maximal, that is,
$\rho = \operatorname{lcm}(\rho_1, \rho_2) \approx 2^{k_1 + k_2}$.
Otherwise, 
the period is 
$\text{$\rho$} = \frac{\rho}{\gcd(s, \rho)}$
\end{proposition}

\begin{proof}
Consider two pseudo-random sequences that are defined by
\begin{equation*}\small
\begin{aligned}
f_i(x) &= a(x)^s\, f_{i-s}(x) \pmod{m_1(x)}, \\[4pt]
g_i(x) &= b(x)^s\, g_{i-s}(x) \pmod{m_2(x)},
\end{aligned}
\end{equation*}
 where $m_1(x)$ and $m_2(x)$ are the two primitive polynomials modulo two where the output is extracted at every $s^{th}$ step.
Then, following Equation~\ref{eq:laurent}, the combined sequence is derived from $V_i(x) = \frac{f_i(x)}{m_1(x)} + \frac{g_i(x)}{m_2(x)} \pmod{1}$.

Consider a polynomial $Z_i(x) = m_2(x)f_i(x) + m_1(x)g_i(x) \pmod{m_1(x)m_2(x)}$. Then, $V_i(x) = \frac{Z_i(x)}{m_1(x)m_2(x)}$.
Now, let $A(x)=n_1(x)m_2(x)(a(x)^s)+n_2(x)m_1(x)(b(x)^s)$
where 
$n_1(x)m_2(x)\equiv 1\pmod{m_1(x)}$ and 
$n_2(x)m_1(x)\equiv 1\pmod{m_2(x)}$.  Then, 
\begin{equation*}
\small
\begin{aligned}
A(x) Z_i(x)&= (n_1(x)m_2(x)(a(x)^s)+n_2(x)m_1(x)(b(x)^s))
              (m_2(x)f_i(x)+m_1(x)g_i(x))\\
           &=n_1(x)m_2^2(x)(a(x)^{2s}) f_{i-s}(x) + n_2(x)m_1^2(x)(b(x)^{2s}) g_{i-s}(x) \pmod{m_1(x)m_2(x)}\\
           &=m_2(x)(a(x)^s) f_i(x) + m_1(x)(b(x)^s) g_i(x) \pmod{m_1(x)m_2(x)}\\ 
           &=m_2(x)f_{i+s}(x)+m_1(x)g_{i+s}(x) \pmod{m_1(x)m_2(x)}\\
           &=Z_{i+s}(x) \pmod{m_1(x)m_2(x)}
\end{aligned}
\end{equation*}
which follows the format of Equation~\ref{eq:laurent}.
Hence, if $s$ is relatively prime to $m_1(x)m_2(x)$, the combined generator achieves the close to maximal period. That is, $\rho = \operatorname{lcm}(\rho_1, \rho_2) \approx 2^{k_1 + k_2}$ if $\gcd{(s, \rho)} =1$; otherwise, it is reduced to $\frac{\rho}{\gcd(s, \rho)}$. Hence, the proof.
\end{proof}

So, if our maximal length CA sizes are relatively prime and the time spacing size is relative prime to the product of the periods of each CA, then we shall achieve the close to maximal period in the combined sequence with time spacing. But, increasing the value of $s$ means more computation is to be done to produce the next number in the sequence and evidently increases the time complexity of the generator. So, to make the generator light-weight, in this study we set $s$ as $2 \leq s \leq 10$. 

Furthermore, this time, we also consider the number of non-zero coefficients $(N_1)$ in the characteristic polynomial of degree $k$ since the fraction $N_1/k$ of non-zero coefficients being close to $50\%$ can be used as a secondary figure of merit for a PRNG \cite{F26}. 
We also increase the maximum CA size to 128, that is taking each component size as $32 \leq k_i \leq 128$. 
Then, we test the equidistribution of this combined CA-based PRNG with time spacing. This type of combined CA-based PRNGs which satisfy the maximally equidistributed condition for certain $s$ values are discussed as follows.


\subsection{Considering at maximum 2 cells with Rule 150}
\label{section4.2.1}

\autoref{table4mca} lists the characteristics of all possible component CAs within degree 32 to 128 from Ref.~\cite{CA4}. In this table, $N_1$ denotes the number of non-zero coefficients and also displays the results of which maximal length CAs in the range s $32 \leq k_i \leq 128$ satisfy secondary figures of merit. Now, we combine CAs for both the case -- when $N_1$ value is close to half of the degree as well as when $N_1$ is not close to half of the degree. 

\subsubsection{Combine maximal length CA with $N_1$ close to half of the degree:}
A review of \autoref{table4mca} indicates that only five maximal length CAs (with CA sizes of 37, 42, 44, 50, and 92) have an $N_1$ value close to half of their degree. From these 5 maximal length CAs, the following two-component relatively prime combinations are possible: $(k_1, k_2) = (37, 42), \ (37, 44),\ (37, 50), \ \text{and}\\ \ (37, 92)$.
All these combinations achieve a large period for the $s$ values in the range $2 \leq s \leq 10$ that are relatively prime to the product of the individual components' period lengths. But, only 3 combinations attain the maximal equidistribution (ME). The following \autoref{table5} describes equidistribution results of these combinations. Here \emph{$s$ period} denotes for which $s$ values close to maximal period is achieved and \emph{$s$ ME} denotes for which $s$ values the generator attains ME. The column \emph{$\rho$ ME} indicates for which $s$ values the generator is both ME as well as attains close to maximal period. For example, for the combined PRNG ($k_1, k_2)$ with $k_1 = 37$ and $k_2 = 42$, the period is equal to $\mathrm{LCM}\left(2^{37} - 1,\, 2^{42}-1\right) \approx 2^{79}$ for $s=2,4,5,8,10$ and it satisfies maximal equidistribution for $s=8,10$. 

\begin{tiny}
\begin{table}[h!]
\vspace{-1.0em}
\centering
\caption{Combined PRNG Equidistribution}
\label{table5}
\resizebox{0.65\textwidth}{!}{
\begin{tabularx}{\textwidth}{|p{1cm}|p{1cm}|>{\centering\arraybackslash}X|>{\centering\arraybackslash}X|>{\centering\arraybackslash}X|}
\hline
$\mathbf{k_1}$ & $\mathbf{k_2}$ & \textbf{$s$ period $(\rho)$} & \textbf{$s$ ME} & \textbf{$\rho$ ME} \\ \hline
37 & 42 &  2, 4, 5, 8, 10  & 8, 10 & 8, 10\\ \hline
37 & 44 &  2, 4, 7, 8      & 7, 8  & 7,8 \\ \hline
37 & 50 &  2, 4, 5, 7, 8, 10 & 10  & 10 \\ \hline
37 & 92 &  2, 4, 7, 8      & -    & - \\ \hline
    \end{tabularx}}
    \vspace{-1.5em}
\end{table}
\end{tiny}

\subsubsection{Combine maximal length CA with $N_1$ not close to half of the degree:}
\label{section4.2.2}
Apart from the five maximal length CAs mentioned above, all other maximal length CAs in \autoref{table4mca} have $N_1$ values that are not close to half of the degree. Next, we identify all two-component combinations whose periods are relatively prime from those CAs; we obtain $2725$ combinations in total. All combinations give a large period for $s$ values relatively prime to $\rho$ as discussed above. Then we test the equidistribution -- out of the $2725$ possible combinations, $368$ combinations attain maximal equidistribution. These results are summarized in the following \autoref{table6equi}. Here, \emph{$s$ period}, \emph{$s$ ME} and \emph{$\rho$ ME} have the same meaning as before. Among these $368$ maximally equidistributed combinations, some combined PRNGs achieve maximal equidistribution and a period close to $2^{k_1+k_2}$ while others achieve maximal equidistribution but the period is not close to the maximal one; in that case, the PRNG has a period of $\frac{\rho}{\gcd(s, \rho)}$. 

\begin{center}
\scriptsize
\setlength{\tabcolsep}{6pt}
\renewcommand{\arraystretch}{0.9}
\vspace{-1.5em}
\begin{longtable}{|c|c|c|c|c||c|c|c|c|c|}
\caption{Equidistribution of Combined CA-based PRNGs with Time Spacing }
\label{table6equi}\\
\hline
\textbf{$k_1$} & \textbf{$k_2$} & \textbf{$s$ period $(\rho)$} & \textbf{$s$ ME} & \textbf{$\rho$ ME} & \textbf{$k_1$} & \textbf{$k_2$} & \textbf{$s$ period $(\rho)$} & \textbf{$s$ ME} & \textbf{$\rho$ ME} \\  \hline
\endfirsthead

\hline
\textbf{$k_1$} & \textbf{$k_2$} & \textbf{$s$ period $(\rho)$} & \textbf{$s$ ME} & \textbf{$\rho$ ME} & \textbf{$k_1$} & \textbf{$k_2$} & \textbf{$s$ period $(\rho)$} & \textbf{$s$ ME} & \textbf{$\rho$ ME} \\  \hline
\endhead

\hline
31&32&2,4,7,8&5,6,7,8,9,10&7,8 & 34&47&2,4,5,7,8,10&7,8,9,10&8,10\\\hline
31&33&2,3,4,5,6,8,9,10&8,9,10&8,9,10 & 34&49&2,4,5,7,8,10&7,8,9,10&8,10\\\hline
31&34&2,4,5,7,8,10&5,6,7,8,9,10&5,7,8,10 & 34&53&2,4,5,7,8,10&9,10&10\\\hline
31&35&2,3,4,5,6,7,8,9,10&8,9,10&8,9,10 & 34&55&2,4,5,7,8,10&8,9,10&8,10\\\hline
31&36&2,4,87,8,9,10&7,8,9,10&8 & 34&57&2,4,5,8,10&9,10&10\\\hline
31&38&2,4,5,7,8,10&7,8,9,10&8,10 & 34&59&2,4,5,7,8,10&10&10\\\hline
31&39&2,3,4,5,6,8,9,10&8,9,10&8,9,10 & 34&61&2,4,5,7,8,10&10&10\\\hline
31&40&2,4,7,8&8,9,10&8 & 34&63&2,4,5,8,10&10&10\\\hline
31&41&2,3,4,5,6,7,8,9,10&8,9,10&8,9,10 & 34&65&2,4,5,7,8,10&10&10\\\hline
31&43&2,3,4,5,6,7,8,9,10&7,9,10&7,9,10 & 35&36&2,4,8&9,10&-\\\hline
31&45&2,3,4,5,6,8,9,10&8,9,10&8,9,10 & 35&38&2,4,5,7,8,10&9,10&10\\\hline
31&46&2,4,5,7,8,10&7,8,9,10&7,10 & 35&43&2,3,4,5,6,7,8,9,10&9,10&9,10\\\hline
31&47&2,3,4,5,6,7,8,9,10&8,9,10&8,9,10 & 35&46&2,4,5,7,8,10&9,10&10\\\hline
31&48&2,4,8&8,9,10&8 & 35&47&2,3,4,5,6,7,8,9,10&8,9,10&8,9,10\\\hline
31&49&2,3,4,5,6,7,8,9,10&8,9,10&8,9,10 & 35&48&2,4,8&7,8,9,10&8\\\hline
31&50&2,4,5,7,8,10&10&8,10 & 35&52&2,4,7,8&9,10&-\\\hline
31&51&2,3,4,5,6,8,9,10&8,9,10&8,9,10 & 35&54&2,4,5,8,10&9,10&10\\\hline
31&52&2,4,7,8&8,9,10&8 & 35&57&2,3,4,5,6,8,9,10&9,10&9,10\\\hline
31&53&2,3,4,5,6,7,8,9,10&9,10&9,10 & 35&58&2,4,5,7,8,10&9,10&10\\\hline
31&54&2,4,5,8,10&9,10&10 & 35&59&2,3,4,5,6,7,8,9,10&9,10&9,10\\\hline
31&55&2,3,4,5,6,7,8,9,10&9,10&9,10 & 35&61&2,3,4,5,6,7,8,9,10&10&10\\\hline
31&57&2,3,4,5,6,8,9,10&9,10&9,10 & 35&62&2,4,5,7,8,10&10&10\\\hline
31&58&2,4,5,7,8,10&9,10&10 & 35&64&2,4,7,8&10&-\\\hline
31&59&2,3,4,5,6,7,8,9,10&10&10 & 35&66&2,4,5,8,10&10&10\\\hline
31&60&2,4,8&10&- & 35&67&2,3,4,5,6,7,8,9,10&10&10\\\hline
31&61&2,3,4,5,6,7,8,9,10&10&10 & 35&68&2,4,7,8&10&-\\\hline
32&33&2,4,8&9,10&- & 36&47&2,4,8&7,8,9,10&8\\\hline
32&43&2,4,7,8&6,7,8,9,10&7,8 & 36&41&2,4,8&9,10&-\\\hline
32&45&2,4,8&7,8,9,10&8 & 36&43&2,4,8&9,10&-\\\hline
32&47&2,4,7,8&7,8,9,10&7,8 & 36&49&2,4,8&9,10&-\\\hline
32&49&2,4,7,8&8,9,10&8 & 36&53&2,4,8&9,10&-\\\hline
32&51&2,4,8&9,10&- & 36&55&2,4,8&8,9,10&8\\\hline
32&53&2,4,7,8&9,10&- & 36&59&2,4,8&9,10&-\\\hline
32&55&2,4,7,8&9,10&- & 36&61&2,4,8&10&-\\\hline
32&57&2,4,8&9,10&- & 36&65&2,4,8&10&-\\\hline
32&59&2,4,7,8&10&- & 36&67&2,4,8&10&-\\\hline
32&61&2,4,7,8&10&- & 37&50&2,4,8&9,10&-\\\hline
32&63&2,4,8&10&- & 37&92&2,4,8&9,10&-\\\hline
33&34&2,4,5,8,10&9,10&10 & 38&39&2,4,5,8,10&10&10\\\hline
33&38&2,4,5,8,10&9,10&10 & 38&41&2,4,5,7,8,10&9,10&10\\\hline
33&40&2,4,8&9,10&- & 38&43&2,4,5,7,8,10&9,10&10\\\hline
33&43&2,3,4,5,6,8,9,10&9,10&9,10 & 38&45&2,4,5,8,10&8,9,10&10\\\hline
33&46&2,4,5,8,10&8,9,10&8,10 & 38&47&2,4,5,7,8,10&9,10&10\\\hline
33&47&2,3,4,5,6,8,9,10&8,9,10&8,9,10 & 38&49&2,4,5,7,8,10&9,10&10\\\hline
33&49&2,3,4,5,6,8,9,10&8,9,10&8,9,10 & 38&51&2,4,5,8,10&9,10&10\\\hline
33&50&2,4,5,8,10&8,9,10&8,10 & 38&53&2,4,5,7,8,10&9,10&10\\\hline
33&52&2,4,8&10&- & 38&55&2,4,5,7,8,10&8,9,10&8,10\\\hline
33&56&2,4,8&10&- & 38&59&2,4,5,7,8,10&9,10&10\\\hline
33&58&2,4,5,8,10&9,10&10 & 38&61&2,4,5,7,8,10&9,10&10\\\hline
33&59&2,3,4,5,6,8,9,10&9,10&9,10 & 38&63&2,4,5,8,10&10&10\\\hline
33&61&2,3,4,5,6,8,9,10&10&10 & 38&65&2,4,5,7,8,10&10&10\\\hline
33&62&2,4,5,8,10&10&10 & 38&67&2,4,5,7,8,10&10&10\\\hline
34&35&2,4,5,7,8,10&9,10&10 & 38&69&2,4,5,8,10&10&10\\\hline
34&39&2,4,5,8,10&9,10&10 & 39&40&2,4,8&9,10&-\\\hline
34&41&2,4,5,7,8,10&9,10&10 & 39&43&2,3,4,5,6,8,9,10&9,10&9,10\\\hline
34&43&2,4,5,7,8,10&8,9,10&8,10 & 39&46&2,4,5,8,10&9,10&10\\\hline
34&45&2,4,5,8,10&7,8,9,10&8,10 & 39&47&2,3,4,5,6,8,9,10&9,10&9,10\\\hline
39 & 49 & 2,3,4,5,6,8,9,10 & 9,10 & 9,10 & 45 & 47 & 2,3,4,5,6,8,9,10 & 8,9,10 & 8,9,10 \\ \hline
39 & 50 & 2,4,5,8,10 & 9,10 & 10 & 45 & 49 & 2,3,4,5,6,8,9,10 & 9,10 & 9,10 \\ \hline
39 & 55 & 2,3,4,5,6,8,9,10 & 9,10 & 9,10 & 45 & 52 & 2,4,8 & 9,10 & - \\ \hline
39 & 56 & 2,4,8 & 9,10 & - & 45 & 53 & 2,3,4,5,6,8,9,10 & 10 & 10 \\ \hline
39 & 58 & 2,4,5,8,10 & 9,10 & 10 & 45 & 56 & 2,4,8 & 9,10 & - \\ \hline
39 & 59 & 2,3,4,5,6,8,9,10 & 9,10 & 9,10 & 45 & 58 & 2,4,5,8,10 & 8,9,10 & 8,10 \\ \hline
39 & 61 & 2,3,4,5,6,8,9,10 & 9,10 & 9,10 & 45 & 59 & 2,3,4,5,6,8,9,10 & 9,10 & 9,10 \\ \hline
39 & 62 & 2,4,5,8,10 & 10 & 10 & 45 & 61 & 2,3,4,5,6,8,9,10 & 9,10 & 9,10 \\ \hline
39 & 67 & 2,3,4,5,6,8,9,10 & 10 & 10 & 45 & 62 & 2,4,5,8,10 & 9,10 & 10 \\ \hline
39 & 68 & 2,4,8 & 10 & - & 45 & 64 & 2,4,8 & 9,10 & - \\ \hline
40 & 41 & 2,4,7,8 & 9,10 & - & 45 & 67 & 2,3,4,5,6,8,9,10 & 10 & 10 \\ \hline
40 & 43 & 2,4,7,8 & 8,9,10 & 8 & 46 & 47 & 2,4,5,7,8,10 & 9,10 & 10 \\ \hline
40 & 47 & 2,4,7,8 & 9,10 & - & 46 & 49 & 2,4,5,7,8,10 & 10 & 10 \\ \hline
40 & 49 & 2,4,7,8 & 9,10 & - & 46 & 51 & 2,4,5,8,10 & 9,10 & 10 \\ \hline
40 & 51 & 2,4,8 & 9,10 & - & 46 & 53 & 2,4,5,7,8,10 & 10 & 10 \\ \hline
40 & 57 & 2,4,8 & 10 & - & 46 & 55 & 2,4,5,7,8,10 & 10 & 10 \\ \hline
40 & 59 & 2,4,7,8 & 9,10 & - & 46 & 57 & 2,4,5,8,10 & 10 & 10 \\ \hline
40 & 61 & 2,4,7,8 & 9,10 & - & 46 & 59 & 2,4,5,7,8,10 & 9,10 & 10 \\ \hline
40 & 63 & 2,4,8 & 10 & - & 46 & 63 & 2,4,5,8,10 & 9,10 & 10 \\ \hline
40 & 67 & 2,4,7,8 & 10 & - & 46 & 65 & 2,4,5,7,8,10 & 9,10 & 10 \\ \hline
41 & 43 & 2,3,4,5,6,7,8,9,10 & 9,10 & 9,10 & 46 & 67 & 2,4,5,7,8,10 & 10 & 10 \\ \hline
41 & 45 & 2,3,4,5,6,8,9,10 & 9,10 & 9,10 & 47 & 48 & 2,4,8 & 9,10 & - \\ \hline
41 & 46 & 2,4,5,7,8,10 & 9,10 & 10 & 47 & 49 & 2,3,4,5,6,7,8,9,10 & 9,10 & 9,10 \\ \hline
41 & 47 & 2,3,4,5,6,7,8,9,10 & 8,9,10 & 8,9,10 & 47 & 50 & 2,4,5,7,8,10 & 9,10 & 10 \\ \hline
41 & 48 & 2,4,8 & 8,10 & 8 & 47 & 51 & 2,3,4,5,6,8,9,10 & 9,10 & 9,10 \\ \hline
41 & 49 & 2,3,4,5,6,7,8,9,10 & 9,10 & 9,10 & 47 & 52 & 2,4,7,8 & 8,10 & 8 \\ \hline
41 & 50 & 2,4,5,7,8,10 & 9,10 & 10 & 47 & 53 & 2,3,4,5,6,7,8,9,10 & 9,10 & 9,10 \\ \hline
41 & 54 & 2,4,5,8,10 & 9,10 & 10 & 47 & 54 & 2,4,5,8,10 & 9,10 & 10 \\ \hline
41 & 55 & 2,3,4,5,6,7,8,9,10 & 9,10 & 9,10 & 47 & 55 & 2,3,4,5,6,7,8,9,10 & 9,10 & 9,10 \\ \hline
41 & 57 & 2,3,4,5,6,8,9,10 & 9,10 & 10 & 47 & 56 & 2,4,7,8 & 8,9,10 & 8 \\ \hline
41 & 58 & 2,4,5,7,8,10 & 9,10 & 10 & 47 & 60 & 2,4,8 & 9,10 & - \\ \hline
41 & 59 & 2,3,4,5,6,7,8,9,10 & 9,10 & 10 & 47 & 64 & 2,4,7,8 & 9,10 & - \\ \hline
41 & 61 & 2,3,4,5,6,7,8,9,10 & 9,10 & 9,10 & 47 & 68 & 2,4,7,8 & 10 & \\ \hline
41 & 62 & 2,4,5,7,8,10 & 10 & 10 & 47 & 72 & 2,4,8 & 10 & - \\ \hline
41 & 63 & 2,3,4,5,6,8,9,10 & 10 & 10 & 47 & 73 & 2,3,4,5,7,8,9,10 & 10 & - \\ \hline
41 & 64 & 2,4,7,8 & 10 & - & 47 & 57 & 2,3,4,5,6,8,9,10 & 9,10 & 9,10 \\ \hline
41 & 66 & 2,4,5,8,10 & 10 & 10 & 47 & 58 & 2,4,5,7,8,10 & 9,10 & 10 \\ \hline
41 & 67 & 2,3,4,5,6,7,8,9,10 & 10 & 10 & 47 & 59 & 2,3,4,5,6,7,8,9,10 & 9,10 & 9,10 \\ \hline
41 & 68 & 2,4,7,8 & 10 & - & 47 & 61 & 2,3,4,5,6,7,8,9,10 & 9,10 & 9,10 \\ \hline
43 & 45 & 2,3,4,5,6,8,9,10 & 9,10 & 9,10 & 47 & 62 & 2,4,5,7,8,10 & 9,10 & 10 \\ \hline
43 & 46 & 2,4,5,7,8,10 & 9,10 & 10 & 47 & 63 & 2,3,4,5,6,8,9,10 & 10 & 10 \\ \hline
43 & 47 & 2,3,4,5,6,7,8,9,10 & 9,10 & 9,10 & 47 & 65 & 2,3,4,5,6,7,8,9,10 & 10 & 10 \\ \hline
43 & 48 & 2,4,8 & 8,9,10 & 8 & 47 & 66 & 2,4,5,8,10 & 10 & 10 \\ \hline
43 & 49 & 2,3,4,5,6,7,8,9,10 & 9,10 & 9,10 & 47 & 67 & 2,3,4,5,6,7,8,9,10 & 10 & 10 \\ \hline
43 & 50 & 2,4,5,7,8,10 & 9,10 & 10 & 47 & 69 & 2,3,4,5,6,8,9,10 & 10 & 10 \\ \hline
43 & 52 & 2,4,7,8 & 7,9,10 & 7 & 47 & 70 & 2,4,5,7,8,10 & 10 & 10 \\ \hline
43 & 55 & 2,3,4,5,6,7,8,9,10 & 9,10 & 9,10 & 47 & 71 & 2,3,4,5,6,7,8,9,10 & 10 & 10 \\ \hline
43 & 56 & 2,4,7,8 & 10 & - & 48 & 49 & 2,4,8 & 9,10 & - \\ \hline
43 & 57 & 2,3,4,5,6,8,9,10 & 9,10 & 9,10 & 48 & 53 & 2,4,8 & 9,10 & - \\ \hline
43 & 58 & 2,4,5,7,8,10 & 8,10 & 8,10 & 48 & 55 & 2,4,8 & 9,10 & - \\ \hline
43 & 59 & 2,3,4,5,6,7,8,9,10 & 9,10 & 9,10 & 48 & 59 & 2,4,8 & 10 & - \\ \hline
43 & 60 & 2,4,8 & 10 & - & 48 & 61 & 2,4,8 & 9,10 & - \\ \hline
43 & 61 & 2,3,4,5,6,7,8,9,10 & 9,10 & 9,10 & 48 & 65 & 2,4,8 & 9,10 & - \\ \hline
43 & 62 & 2,4,5,7,8,10 & 10 & 10 & 48 & 71 & 2,4,8 & 10 & - \\ \hline
43 & 63 & 2,3,4,5,6,8,9,10 & 9,10 & 9,10 & 48 & 73 & 2,4,8 & 10 & - \\ \hline
43 & 66 & 2,4,5,8,10 & 10 & 10 & 48 & 77 & 2,4,8 & 10 & - \\ \hline
43 & 67 & 2,3,4,5,6,7,8,9,10 & 10 & 10 & 49 & 50 & 2,4,5,7,8,10 & 9,10 & 10 \\ \hline
43 & 68 & 2,4,7,8 & 10 & - & 49 & 51 & 2,3,4,5,6,8,9,10 & 10 & 10 \\ \hline 

49 & 52 & 2,4,7,8 & 9,10 & - & 55 & 56 & 2,4,7,8 & 9,10 & - \\ \hline
49 & 60 & 2,4,8 & 9,10 & - & 55 & 57 & 2,3,4,5,6,8,9,10 & 9,10 & 9,10 \\ \hline
49 & 68 & 2,4,7,8 & 10 & - & 55 & 59 & 2,3,4,5,6,7,8,9,10 & 10 & 10 \\ \hline
49 & 53 & 2,3,4,5,6,7,8,9,10 & 10 & 10 & 55 & 61 & 2,3,4,5,6,7,8,9,10 & 10 & 10 \\ \hline
49 & 54 & 2,4,5,8,10 & 10 & 10 & 55 & 62 & 2,4,5,7,8,10 & 10 & 10 \\ \hline
49 & 55 & 2,3,4,5,6,7,8,9,10 & 9,10 & 9,10 & 55 & 63 & 2,3,4,5,6,8,9,10 & 9,10 & 9,10 \\ \hline
49 & 57 & 2,3,4,5,6,8,9,10 & 10 & 10 & 55 & 64 & 2,4,7,8 & 9,10 & - \\ \hline
49 & 58 & 2,4,5,7,8,10 & 8,9,10 & 8,10 & 55 & 68 & 2,4,7,8 & 10 & - \\ \hline
49 & 59 & 2,3,4,5,6,7,8,9,10 & 9,10 & 9,10 & 55 & 69 & 2,3,4,5,6,8,9,10 & 10 & 10 \\ \hline
49 & 62 & 2,4,5,7,8,10 & 10 & 10 & 56 & 57 & 2,4,8 & 10 & - \\ \hline
49 & 65 & 2,3,4,5,6,7,8,9,10 & 10 & 10 & 56 & 59 & 2,4,7,8 & 10 & - \\ \hline
49 & 66 & 2,4,5,8,10 & 9,10 & 10 & 56 & 61 & 2,4,7,8 & 10 & - \\ \hline
49 & 67 & 2,3,4,5,6,7,8,9,10 & 10 & 10 & 56 & 67 & 2,4,7,8 & 10 & - \\ \hline
49 & 69 & 2,3,4,5,6,8,9,10 & 10 & 10 & 56 & 71 & 2,4,7,8 & 10 & - \\ \hline
50 & 51 & 2,4,5,8,10 & 9,10 & 10 & 57 & 58 & 2,4,5,8,10 & 9,10 & 10 \\ \hline
50 & 53 & 2,4,5,7,8,10 & 9,10 & 10 & 57 & 59 & 2,3,4,5,6,8,9,10 & 10 & 10 \\ \hline
50 & 57 & 2,4,5,8,10 & 9,10 & 10 & 57 & 61 & 2,3,4,5,6,8,9,10 & 10 & 10 \\ \hline
50 & 59 & 2,4,5,7,8,10 & 9,10 & 10 & 57 & 62 & 2,4,5,8,10 & 10 & 10 \\ \hline
50 & 61 & 2,4,5,7,8,10 & 9,10 & 10 & 57 & 70 & 2,4,5,8,10 & 10 & 10 \\ \hline
50 & 63 & 2,4,5,8,10 & 10 & 10 & 57 & 77 & 2,3,4,5,6,8,9,10 & 10 & 10 \\ \hline
50 & 67 & 2,4,5,7,8,10 & 10 & 10 & 58 & 59 & 2,4,5,7,8,10 & 10 & 10 \\ \hline
50 & 69 & 2,4,5,8,10 & 10 & 10 & 58 & 61 & 2,4,5,7,8,10 & 9,10 & 10 \\ \hline
50 & 71 & 2,4,5,7,8,10 & 10 & 10 & 58 & 63 & 2,4,5,8,10 & 10 & 10 \\ \hline
50 & 73 & 2,4,5,7,8,10 & 10 & 10 & 58 & 65 & 2,4,5,7,8,10 & 10 & 10 \\ \hline
51 & 55 & 2,3,4,5,6,8,9,10 & 9,10 & 9,10 & 58 & 69 & 2,4,5,8,10 & 10 & 10 \\ \hline
51 & 56 & 2,4,8 & 9,10 & - & 58 & 75 & 2,4,5,8,10 & 10 & 10 \\ \hline
51 & 58 & 2,4,5,8,10 & 9,10 & 10 & 58 & 77 & 2,4,5,7,8,10 & 10 & 10 \\ \hline
51 & 59 & 2,3,4,5,6,8,9,10 & 9,10 & 9,10 & 58 & 81 & 2,4,5,8,10 & 10 & 10 \\ \hline
51 & 67 & 2,3,4,5,6,8,9,10 & 9,10 & 9,10 & 59 & 61 & 2,3,4,5,6,7,8,9,10 & 10 & 10 \\ \hline
51 & 71 & 2,3,4,5,6,8,9,10 & 10 & 10 & 59 & 60 & 2,4,8 & 9,10 & - \\ \hline
52 & 53 & 2,4,7,8 & 9,10 & 8 & 59 & 62 & 2,4,5,7,8,10 & 10 & 10 \\ \hline
52 & 55 & 2,4,7,8 & 8,9,10 & - & 59 & 63 & 2,3,4,5,6,8,9,10 & 9,10 & 9,10 \\ \hline
52 & 57 & 2,4,8 & 9,10 & - & 59 & 64 & 2,4,7,8 & 10 & 10 \\ \hline
52 & 59 & 2,4,7,8 & 9,10 & - & 59 & 65 & 2,3,4,5,6,7,8,9,10 & 10 & 10 \\ \hline
52 & 61 & 2,4,7,8 & 9,10 & - & 59 & 66 & 2,4,5,8,10 & 10 & 10 \\ \hline
52 & 63 & 2,4,8 & 10 & - & 59 & 70 & 2,4,5,7,8,10 & 10 & 10 \\ \hline
52 & 67 & 2,4,7,8 & 10 & - & 59 & 72 & 2,4,8 & 10 & 10 \\ \hline
52 & 69 & 2,4,8 & 10 & - & 59 & 76 & 2,4,7,8 & 10 & 10 \\ \hline
52 & 71 & 2,4,7,8 & 10 & - & 59 & 77 & 2,3,4,5,6,7,8,9,10 & 10 & 10 \\ \hline
53 & 54 & 2,4,5,8,10 & 10 & 10 & 59 & 78 & 2,4,5,8,10 & 10 & 10 \\ \hline
53 & 55 & 2,3,4,5,6,7,8,9,10v & 910 & 9,10 & 60 & 61 & 2,4,8 & 9,10 & - \\ \hline
53 & 56 & 2,4,7,8 & 9,10 & - & 60 & 67 & 2,4,8 & 10 & - \\ \hline
53 & 57 & 2,3,4,5,6,8,9,10 & 10 & 10 & 60 & 71 & 2,4,8 & 10 & - \\ \hline
53 & 58 & 2,4,5,7,8,10 & 9,10 & 10 & 60 & 73 & 2,4,8 & 10 & - \\ \hline
53 & 59 & 2,3,4,5,6,7,8,9,10 & 9,10 & 9,10 & 61 & 63 & 2,3,4,5,6,8,9,10 & 10 & 10 \\ \hline
53 & 60 & 2,4,8 & 9,10 & - & 61 & 66 & 2,4,5,8,10 & 10 & 10 \\ \hline
53 & 61 & 2,3,4,5,6,7,8,9,10 & 10 & 10 & 61 & 77 & 2,3,4,5,6,7,8,9,10 & 10 & 10 \\ \hline
53 & 62 & 2,4,5,7,8,10 & 10 & 10 & 61 & 78 & 2,4,5,8,10 & 10 & 10 \\ \hline
53 & 63 & 2,3,4,5,6,8,9,10 & 10 & 10 & 62 & 63 & 2,4,5,8,10 & 10 & 10 \\ \hline
53 & 66 & 2,4,5,8,10 & 9,10 & 10 & 62 & 77 & 2,4,5,7,8,10 & 10 & 10 \\ \hline
53 & 67 & 2,3,4,5,6,7,8,9,10 & 10 & 10 & 63 & 65 & 2,3,4,5,6,8,9,10 & 10 & 10 \\ \hline
53 & 71 & 2,3,4,5,6,7,8,9,10 & 10 & 10 & 63 & 67 & 2,3,4,5,6,8,9,10 & 10 & 10 \\ \hline
54 & 55 & 2,4,5,8,10 & 9,10 & 10 & 63 & 71 & 2,3,4,5,6,8,9,10 & 10 & 10 \\ \hline
54 & 59 & 2,4,5,8,10 & 9,10 & 10 & 63 & 73 & 2,3,4,5,6,8,9,10 & 10 & 10 \\ \hline
54 & 61 & 2,4,5,8,10 & 10 & 10 & 63 & 74 & 2,4,5,8,10 & 10 & 10 \\ \hline
54 & 67 & 2,4,5,8,10 & 10 & 10 & 63 & 79 & 2,3,4,5,6,8,9,10 & 10 & 10 \\ \hline
54 & 71 & 2,4,5,8,10 & 10 & 10 & 63 & 80 & 2,4,8 & 10 & - \\ \hline
64 & 71 & 2,4,7,8 & 10 & - & 67 & 72 & 2,4,8 & 10 & 10 \\ \hline
64 & 77 & 2,4,7,8 & 10 & - & 67 & 77 & 2,3,4,5,6,7,8,9,10 & 10 & 10 \\ \hline
65 & 66 & 2,4,5,8,10 & 10 & 10 & 69 & 71 & 2,3,4,5,6,8,9,10 & 10 & 10 \\ \hline
65 & 71 & 2,3,4,5,6,7,8,9,10 & 10 & 10 & 70 & 71 & 2,4,5,7,8,10 & 10 & 10 \\ \hline
65 & 76 & 2,4,7,8 & 10 & - & 70 & 73 & 2,4,5,7,8,10 & 10 & 10 \\ \hline
65 & 77 & 2,3,4,5,6,7,8,9,10 & 10 & 10 & 70 & 79 & 2,4,5,7,8,10 & 10 & 10 \\ \hline
66 & 73 & 2,4,5,8,10 & 10 & 10 & 71 & 74 & 2,4,5,7,8,10 & 10 & 10 \\ \hline
66 & 79 & 2,4,5,8,10 & 10 & 10 & 71 & 72 & 2,4,8 & 10 & - \\ \hline
66 & 83 & 2,4,5,8,10 & 10 & 10 & 71 & 75 & 2,3,4,5,6,8,9,10 & 10 & 10 \\ \hline
67 & 68 & 2,4,7,8 & 10 & - & 71 & 77 & 2,3,4,5,6,7,8,9,10 & 10 & 10 \\ \hline
67 & 69 & 2,3,4,5,6,8,9,10 & 10 & - & 71 & 78 & 2,4,5,8,10 & 10 & 10 \\ \hline
67 & 70 & 2,4,5,7,8,10 & 10 & 10 & & & & & \\ \hline
\end{longtable}
\vspace{-2.5em}
\end{center}


\begin{example}
 Take the same combined PRNG as discussed in Example~\ref{ex:1}. Here, to improve equidistribution, we combine the CA configurations based on the $s$ values. This combined PRNG with components of $k_1$=31 and $k_2$=32 achieves the maximal period when $s$=\{2, 4, 7, 8\}. However, when tested for maximal equidistribution, \autoref{table6} to \autoref{table9} show that, out of the 4 $s$ values, it gives maximally equidistribution for only $s$=7 and 8. So this combined PRNG (31,32) gives a period close to $2^{63}$ and maximal equidistribution for $s$=7 and 8. 


\begin{scriptsize}
\vspace{-1.0em}
\begin{longtable}{|>{\centering\arraybackslash}p{1cm}|
                        >{\centering\arraybackslash}p{3cm}|
                        >{\centering\arraybackslash}p{2cm}|
                        >{\centering\arraybackslash}p{2cm}|
                        >{\centering\arraybackslash}p{4cm}|}
\caption{Equidistribution of combined CA-based PRNG (31,32) for $s$=2\label{table6}}\\
\hline
\textbf{t} & \textbf{$l_t^* = \min(L, \lfloor k/t \rfloor)$} & \textbf{$l_t \ (l_t \le l_t^*)$} & \textbf{Rank} & \textbf{Equidistribution} \\
\hline
\endfirsthead

\hline
\textbf{t} & \textbf{$l_t^* = \min(L, \lfloor k/t \rfloor)$} & \textbf{$l_t \ (l_t \le l_t^*)$} & \textbf{Rank} & \textbf{Equidistribution } \\
\hline
\endhead

\hline 
\endfoot

\hline
\endlastfoot

2  & 31 & 31 & 47 & not $(t,l)$-equi-distributed \\
3  & 21 & 21 & 46 & not $(t,l)$-equi-distributed \\
4  & 15 & 15 & 41 & not $(t,l)$-equi-distributed \\
5  & 12 & 12 & 38 & not $(t,l)$-equi-distributed \\
6  & 10 & 10 & 36 & not $(t,l)$-equi-distributed \\
7  & 9  & 9  & 37 & not $(t,l)$-equi-distributed \\
9  & 7  & 7  & 43 & not $(t,l)$-equi-distributed \\
10 & 6  & 6  & 46 & not $(t,l)$-equi-distributed \\
12 & 5  & 5  & 53 & not $(t,l)$-equi-distributed \\
15 & 4  & 4  & 60 & $(t,l)$-equi-distributed \\
21 & 3  & 3  & 63 & $(t,l)$-equi-distributed \\
31 & 2  & 2  & 62 & $(t,l)$-equi-distributed \\
63 & 1  & 1  & 63 & $(t,l)$-equi-distributed \\

\end{longtable}
\vspace{-1.0em}
\end{scriptsize}

\begin{scriptsize}
\vspace{-1.0em}

\begin{longtable}{|>{\centering\arraybackslash}p{1cm}|
                        >{\centering\arraybackslash}p{3cm}|
                        >{\centering\arraybackslash}p{2cm}|
                        >{\centering\arraybackslash}p{2cm}|
                        >{\centering\arraybackslash}p{4cm}|}
\caption{Equidistribution of combined CA-based PRNG (31,32) for $s$=4\label{table7}}\\
\hline
\textbf{t} & \textbf{$l_t^* = \min(L, \lfloor k/t \rfloor)$} & \textbf{$l_t \ (l_t \le l_t^*)$} & \textbf{Rank} & \textbf{Equidistribution} \\
\hline
\endfirsthead

\hline
\textbf{t} & \textbf{$l_t^* = \min(L, \lfloor k/t \rfloor)$} & \textbf{$l_t \ (l_t \le l_t^*)$} & \textbf{Rank} & \textbf{Equidistribution $(t,l_t)$} \\
\hline
\endhead

\hline 
\endfoot

\hline
\endlastfoot

2  & 31 & 31 & 60 & not $(t,l)$-equi-distributed \\
3  & 21 & 21 & 59 & not $(t,l)$-equi-distributed \\
4  & 15 & 15 & 57 & not $(t,l)$-equi-distributed \\
5  & 12 & 12 & 54 & not $(t,l)$-equi-distributed \\
6  & 10 & 10 & 58 & not $(t,l)$-equi-distributed \\
7  & 9  & 9  & 62 & not $(t,l)$-equi-distributed \\
9  & 7  & 7  & 63 & $(t,l)$-equi-distributed \\
10 & 6  & 6  & 60 & $(t,l)$-equi-distributed \\
12 & 5  & 5  & 60 & $(t,l)$-equi-distributed \\
15 & 4  & 4  & 60 & $(t,l)$-equi-distributed \\
21 & 3  & 3  & 63 & $(t,l)$-equi-distributed \\
31 & 2  & 2  & 62 & $(t,l)$-equi-distributed \\
63 & 1  & 1  & 63 & $(t,l)$-equi-distributed \\

\end{longtable}
\vspace{-1.0em}

\end{scriptsize}

\begin{scriptsize}
\vspace{-1.0em}

\begin{longtable}{|>{\centering\arraybackslash}p{1cm}|
                        >{\centering\arraybackslash}p{3cm}|
                        >{\centering\arraybackslash}p{2cm}|
                        >{\centering\arraybackslash}p{2cm}|
                        >{\centering\arraybackslash}p{4cm}|}
\caption{Equidistribution of combined CA-based PRNG (31,32) for $s$=7\label{table8}}\\
\hline
\textbf{t} & \textbf{$l_t^* = \min(L, \lfloor k/t \rfloor)$} & \textbf{$l_t \ (l_t \le l_t^*)$} & \textbf{Rank} & \textbf{Equidistribution} \\
\hline
\endfirsthead

\hline
\textbf{t} & \textbf{$l_t^* = \min(L, \lfloor k/t \rfloor)$} & \textbf{$l_t \ (l_t \le l_t^*)$} & \textbf{Rank} & \textbf{Equidistribution $(t,l_t)$} \\
\hline
\endhead

\hline 
\endfoot

\hline
\endlastfoot

2  & 31 & 31 & 62 & $(t,l)$-equi-distributed \\
3  & 21 & 21 & 63 & $(t,l)$-equi-distributed \\
4  & 15 & 15 & 60 & $(t,l)$-equi-distributed \\
5  & 12 & 12 & 60 & $(t,l)$-equi-distributed \\
6  & 10 & 10 & 60 & $(t,l)$-equi-distributed \\
7  & 9  & 9  & 63 & $(t,l)$-equi-distributed \\
9  & 7  & 7  & 63 & $(t,l)$-equi-distributed \\
10 & 6  & 6  & 60 & $(t,l)$-equi-distributed \\
12 & 5  & 5  & 60 & $(t,l)$-equi-distributed \\
15 & 4  & 4  & 60 & $(t,l)$-equi-distributed \\
21 & 3  & 3  & 63 & $(t,l)$-equi-distributed \\
31 & 2  & 2  & 62 & $(t,l)$-equi-distributed \\
63 & 1  & 1  & 63 & $(t,l)$-equi-distributed \\

\end{longtable}
\vspace{-1.0em}

\end{scriptsize}

\begin{scriptsize}
\vspace{-1.0em}

\begin{longtable}{|>{\centering\arraybackslash}p{1cm}|
                        >{\centering\arraybackslash}p{3cm}|
                        >{\centering\arraybackslash}p{2cm}|
                        >{\centering\arraybackslash}p{2cm}|
                        >{\centering\arraybackslash}p{4cm}|}
\caption{Equidistribution of combined CA-based PRNG (31,32) for $s$=8\label{table9}}\\
\hline
\textbf{t} & \textbf{$l_t^* = \min(L, \lfloor k/t \rfloor)$} & \textbf{$l_t \ (l_t \le l_t^*)$} & \textbf{Rank} & \textbf{Equidistribution } \\
\hline
\endfirsthead

\multicolumn{5}{c}{{\tablename\ \thetable{} -- }} \\
\hline
\textbf{t} & \textbf{$l_t^* = \min(L, \lfloor k/t \rfloor)$} & \textbf{$l_t \ (l_t \le l_t^*)$} & \textbf{Rank} & \textbf{Equidistribution } \\
\hline
\endhead

\hline 
\endfoot

\hline
\endlastfoot

2  & 31 & 31 & 62 & $(t,l)$-equi-distributed \\
3  & 21 & 21 & 63 & $(t,l)$-equi-distributed \\
4  & 15 & 15 & 60 & $(t,l)$-equi-distributed \\
5  & 12 & 12 & 60 & $(t,l)$-equi-distributed \\
6  & 10 & 10 & 60 & $(t,l)$-equi-distributed \\
7  & 9  & 9  & 63 & $(t,l)$-equi-distributed \\
9  & 7  & 7  & 63 & $(t,l)$-equi-distributed \\
10 & 6  & 6  & 60 & $(t,l)$-equi-distributed \\
12 & 5  & 5  & 60 & $(t,l)$-equi-distributed \\
15 & 4  & 4  & 60 & $(t,l)$-equi-distributed \\
21 & 3  & 3  & 63 & $(t,l)$-equi-distributed \\
31 & 2  & 2  & 62 & $(t,l)$-equi-distributed \\
63 & 1  & 1  & 63 & $(t,l)$-equi-distributed \\

\end{longtable}
\vspace{-1.0em}
\end{scriptsize}
\end{example}
Consider another combined PRNG $(k_1,k_2) = (59,64)$; it achieves close to maximal period for $s$ values 2, 4, 8 and maximal equidistribution for $s$=10. So this combined PRNG (59,64) with $s$=10 is maximally equidistributed but the period is $\frac{\rho}{\gcd(10, \rho)}$ $\approx 2^{121}$ (not close to maximal) because $\gcd(\rho, 10) \neq 1$. Additionally, we found $1399$ combinations which are almost maximally equidistributed as well as have a period close to $2^{k_1+k_2}$. We only experimented for $s$ values from 2 to 10; if we increase the values of $s$, it increases the computational time, but we might get some more maximally equidistributed combined generators. 

\subsection{Considering $CA(90')$ or $CA(150')$}
\label{section4.3}
Next, we combine the $CA(90')$ and $CA(150')$ close to 32, 64, and 128 with time spacing. Here, only one CA has $N_1$ close to half of its degree and the remaining ones are not close to half of the degree. Therefore, we combine all the two-component relatively prime combinations close to 32, 64 and 128 given in \autoref{table4}.




Here also, all the combined PRNG combinations (close to 32, 64, and 128) achieve the maximal period for those $s$ values if $\gcd(s, \rho) = 1$. But, when tested for the equidistribution, the combined CA-based PRNGs in the form of $CA(90')$ with $CA(90')$ and  $CA(150')$ with $CA(150')$ do not satisfy the maximally equidistributed property for CAs close to 64 and 128 but are almost maximally equidistributed with maximal period for CAs close to 32. However, the generator combining $CA(90')$ with $CA(150')$ satisfies the maximal equidistribution and maximal period for CAs close to 32 and 64 and almost ME and as well as maximal period for CAs close to 128. These results are summarized in the  following \autoref{table11}.

\begin{scriptsize}
\renewcommand{\arraystretch}{1.5} 
\setlength{\tabcolsep}{8pt}       
\vspace{-1.0em}

\begin{longtable}{|>{\centering\arraybackslash}p{2cm}|
                >{\raggedright\arraybackslash}p{2cm}|
                >{\raggedright\arraybackslash}p{2cm}|
                >{\raggedright\arraybackslash}p{5.5cm}|}
\caption{Equidistribution of Combined PRNG using $CA(90')$ and $ CA(150')$}
\label{table11}\\
\hline
\multirow{2}{*}{\textbf{CA}} & 
\multicolumn{3}{c|}{\textbf{Combined PRNG with maximal period and  ME (or, almost ME)}} \\ \cline{2-4}
& \textbf{\scriptsize{$(k_1,k_2,s)$ : $(90',90',s)$}} &
\textbf{\scriptsize{$(k_1,k_2,s)$ : $(150',150',s)$}}&
\textbf{$(k_1,k_2,s)$ : $(90',150',s)$} \\ \hline

Close to 32 &
\textbf{almost ME}:
(26,29,10), (26,35,10), (29,35,10) &
\textbf{almost ME}:
(26,29,10), (26,35,10), (29,35,10) &
\textbf{ME:}
(26,29,\{5,7,8,10\}), (29,26,\{5,7,8,10\}), (26,35,\{5,7,8,10\}), (35,26,\{5,7,8,10\}), 
(29,39,\{6,8,9,10\}), (39,29,\{6,8,9,10\}), (35,39,\{6,8,9,10\}), (39,35,\{6,8,9,10\}),
(29,35,\{5,6,7,8,9,10\}), (35,29,\{5,6,7,8,9,10\}) \\ \hline

Close to 64 &
not ME &
not ME &
\textbf{ME:}
(65,69,10), (69,65,10) \\ \hline

Close to 128 &
not ME &
not ME &
\textbf{almost ME}:
(105,113,\{9,10\}), (113,105,\{9,10\}), (113,119,10), (119,113,10) \\ \hline

\end{longtable}
\vspace{-1.0em}

\end{scriptsize}
In the next section, we test these (almost or) maximally equidistributed combined generators in statistical testbeds.

\section{Experimental Results and Analysis}
\label{section5}
This section describes the experimental results for the proposed combined CA-based PRNGs. The analysis focuses on the following aspects: empirical test results, verification with space-time diagrams and performance (speed). 

\subsection{Empirical Test}
\label{section5.1}

Empirical tests are used to assess the randomness quality of PRNGs. There are several empirical tests available for testing PRNGs, including the \emph{Dieharder} \cite{dieharder27}, \emph{SmallCrush}, and \emph{BigCrush} \cite{testu0128} tests.
In this work, we first test the combined maximal length CA-based PRNGs with time spacing, which achieve the maximally equidistributed characteristic. So, our experiment begins by testing the PRNGs in Dieharder by generating pseudo-random numbers for all the combinations and the corresponding $s$ values mentioned in \autoref{table5} and \autoref{table6equi} with a binary file size of 1.5 GB.

We observe that, the combined maximal length CA-based PRNGs where $N_1$ is close to half of the degree (\autoref{table5}) fail most of the tests in Dieharder. Therefore, they are not further tested in SmallCrush and BigCrush tests. So, this type of combined CA-based PRNGs may achieve the maximal period and maximal equidistribution but fail in statistical testbeds.

We also test the combined PRNGs, which are almost ME on Dieharder. There are a total of 1399 combinations that achieve almost maximal equidistribution. From these, we take the combined PRNGs with minimal skip step; those are $(31,33,5), (31,35,5), (31,38,5), (31,39,5), (31,41,5), (31,43,5), (31,\\46,5), (31,49,5), (31,51,5), (33,43,5)$, and $(33,47,5)$. But all these PRNGs failed most of the tests in Dieharder and so are not tested further.

Then we test the combined CA-based PRNGs using $CA(90')$ and $CA(150')$, which satisfy the maximal and almost maximal equidistribution as mentioned in \autoref{table11}. We observe that, even though the combined CA-based PRNG with time spacing using $CA(90')$ and $CA(150')$ can provide period close to the maximal one, satisfies ME (or, almost ME), but these PRNGs do not provide any better statistical test results when we test with Dieharder.

\begin{scriptsize}
\setlength{\tabcolsep}{2pt}
\renewcommand{\arraystretch}{0.9}
\begin{longtable}{|c|c|c|c|c|c|}
\caption{Period Length, Equidistribution, and Statistical test results of Linear PRNGs}
\label{table14} \\
\hline
\textbf{PRNG} & \textbf{Period ($\approx \rho$)} & \textbf{Equidistribution} & \makecell{\textbf{Dieharder}} & \makecell{\textbf{SmallCrush}} & \makecell{\textbf{BigCrush}} \\
\hline
\endfirsthead

\hline
\textbf{PRNGs} & \textbf{Period ($\approx \rho$)} & \textbf{Equidistribution} & \makecell{\textbf{Dieharder}} & \makecell{\textbf{SmallCrush}} & \makecell{\textbf{BigCrush}} \\
\hline
\endhead

\multicolumn{6}{|c|}{\textbf{Combined CAs-based PRNGs $(k_1, k_2, s)$ ($\rho$ close to maximal)}} \\
\hline
 $(31,32,7)$   & $2^{63}$  & ME        & 1 & 1 & 11 \\
 $(31,32,8)$   & $2^{63}$  & ME        & 1 & 1 & 7 \\
 $(31,40,8)$   & $2^{71}$  & ME        & 7 & 3 & 16 \\
 $(35,48,8)$   & $2^{83}$  & ME        & 8 & 4 & 19 \\
 $(41,48,8)$   & $2^{89}$  & ME        & 9 & 2 & 13 \\
 $(43,48,8)$   & $2^{91}$  & ME        & 5 & 2 & 11 \\
 $(47,56,8)$   & $2^{103}$ & ME        & 4 & 2 & 13 \\
\hline

\multicolumn{6}{|c|}{\textbf{Combined CAs-based PRNG $(k_1, k_2, s)$ ($\rho$ not close to maximal)}} \\
\hline
 $(31,32,5)$  &  $2^{61}$  &  ME  & 3 & 3 & 20 \\
 $(31,32,6)$  &  $2^{62}$  &  ME  & 2 & 1 & 15 \\
 $(31,32,9)$  &  $2^{62}$  &  ME  & All passed & 1 & 6 \\
 $(31,32,10)$ &  $2^{61}$  &  ME & All passed & 2 & 6 \\ 
 $(31,40,9)$  &  $2^{70}$  & ME  & All passed & 2 & 16 \\
 $(31,40,10)$ & $2^{69}$   & ME  & All passed & 2 & 13 \\
 $(33,40,9)$ & $2^{72}$ & ME & All passed & 1 & 15 \\
 $(33,40,10)$ & $2^{71}$ & ME & All passed & 1 & 14 \\
 $(33,56,10)$ & $2^{87}$ & ME & 1 & 3 & 12 \\
 $(35,48,7)$   & $2^{81}$  & ME  & 10 & 5 & 20 \\
 $(35,48,9)$   & $2^{80}$  & ME  & 1 & 2 & 13 \\
 $(35,48,10)$   & $2^{81}$  & ME & 1 & 1 & 13 \\
 $(35,64,10)$ & $2^{97}$ & ME & 1 & 3 & 11 \\
 $(39,40,9)$ & $2^{78}$ & ME & 1 & 1 & 13 \\
 $(39,40,10)$ & $2^{77}$ & ME & 1 & 1 & 11  \\
 $(39,56,9)$ & $2^{94}$ & ME & 1 & 2 & 11 \\
 $(39,56,10)$ & $2^{93}$ & ME & 1 & 1 & 10 \\
 $(41,48,10)$   & $2^{87}$  & ME & All passed & 1 & 13 \\
 $(41,64,10)$ & $2^{103}$ & ME & 1 & 3 & 12 \\
 $(43,48,9)$   & $2^{88}$  & ME  & 1 & 1 & 12 \\
 $(43,48,10)$   & $2^{89}$  & ME & All passed  & 1 & 10 \\
 $(43,56,10)$ & $2^{97}$ & ME & 1 & 1 & 10 \\
 $(45,56,9)$ & $2^{100}$ & ME & All passed & 1 & 11 \\
 $(45,56,10)$ & $2^{99}$ & ME & All passed  & 1 & 10 \\
 $(45,64,9)$ & $2^{108}$ & ME & 2 & 3 & 13 \\
 $(45,64,10)$ & $2^{107}$ & ME & 2 & 2 & 10 \\
 $(47,56,9)$   & $2^{102}$ & ME  & All passed  & 1 & 11 \\
 $(47,56,10)$   & $2^{101}$ & ME & All passed  & 1 & 11 \\
 $(47,64,9)$ & $2^{110}$ & ME & 1 & 2 & 12 \\
 $(47,64,10)$ & $2^{109}$ & ME & 1 & 2 & 10 \\
 $(47,72,10)$ & $2^{117}$ & ME & 1 & 2 & 10 \\
 $(51,56,9)$ & $2^{106}$ & ME & All passed & 1 & 12 \\
 $(51,56,10)$ & $2^{105}$ & ME & All passed & 1 & 10  \\
 $(53,56,9)$ & $2^{108}$ & ME & All passed & 1 & 11 \\
 $(53,56,10)$ & $2^{107}$ & ME & All passed  & 1 & 11  \\
 $(55,56,9)$ & $2^{110}$ & ME & All passed & 1 & 7 \\
 $(55,56,10)$ & $2^{109}$ & ME & All passed & 1 & 7  \\
 $(59,64,10)$ & $2^{121}$ & ME & All passed & 1 & 11 \\
 $(63,64,10)$ & $2^{125}$ & ME & 1 & 1 & 13 \\
 $(63,80,10)$ & $2^{141}$ & ME & All passed & 3 & 11\\
 ${(67,72,10)}$ & $2^{137}$ & ME & All passed & All passed & 7 \\
 $(71,72,10)$ & $2^{141}$ & ME & All passed & 2 & 9 \\
\hline
\multicolumn{6}{|c|}{\textbf{Existing Linear PRNGs}} \\
\hline
Tausworthe (combined) & $2^{88}$ & ME & All passed & All passed & 7 \\
Mersenne Twister       & $2^{19937}-1$ & not ME & All passed & All passed & 3 \\
WELL512a               & $2^{512}$ & ME & All passed & All passed & 6 \\
WELL1024a              & $2^{1024}$ & ME & All passed & All passed & 7 \\
\hline

\end{longtable}
\vspace{-1.0em}
\end{scriptsize}

Finally, we test the combined maximal length CA-based PRNGs of \autoref{table6equi} where $N_1$ is not close to half of the degree in Dieharder. There are a total of 368 combinations that achieve the maximal equidistribution. Out of these combinations, the following combined PRNGs passed the majority of tests (above 100 tests) in Dieharder: $(k_1, k_2, s) = (31,32,\{5,6,7,8,9,10\}), (31,40,\\\{8,9,10\}), (33,40,\{9,10\}), (33,56,\{10\}), (35, 48,\{7,8,9,10\}), (35,64,\{10\}), \\(39,40,\{9,10\}), (39,56,\{9,10\}), (41,48,\{8,10\}),(41,64,\{10\}), (43,48,\{8,9,\\10\}), (43,56,\{9,10\}), (45,56,\{9,10\}), (45,64,\{9,10\}), (47,56,\{8,9,10\}), (47,\\64,\{9,10\}), (47,72,\{10\}), (51,56,\{9,10\}), (53,56,\{9,10\}), (55,56,\{9,10\}), \\(59,64,\{10\}), (63,64,\{10\}), (63,80,\{10\}), (67,72,\{10\}), (71,72,\{10\})$. Some of these combined PRNGs achieve the maximal equidistribution and are close to the maximal period of $2^{k_1+k_2}$, and some are maximally equidistributed and have the period of $\frac{\rho}{\gcd(s, \rho)}$. So, these combined PRNGs are tested with SmallCrush and BigCrush tests. The results are shown in \autoref{table14}. These results are then compared with the well-known linear PRNGs like Tausworthe (combined), Mersenne Twister, WELL512a and WELL1024a. It can be observed from \autoref{table14} that, many of our proposed generators have performed at par with these standard generators showing their efficacy as PRNGs.



\begin{figure}[!hbtp]
\vspace{-4.5em}
    \centering
    \begin{subfigure}[b]{0.3\textwidth}
        \centering
        \fbox{%
        \includegraphics[width=\textwidth]{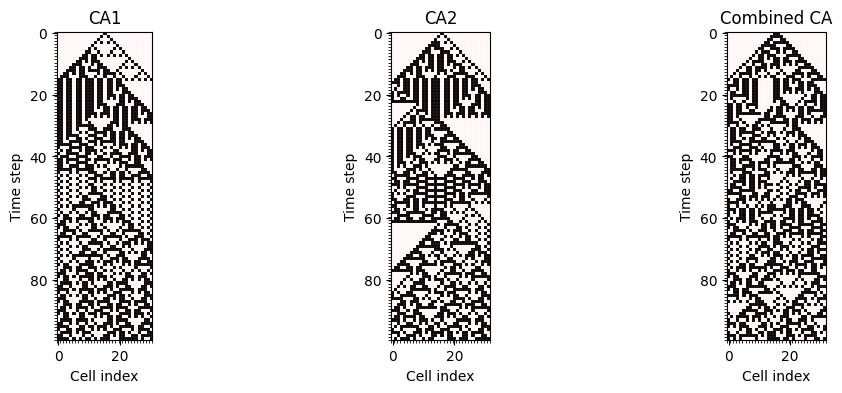}
        }
        \caption{$(k_1,k_2):(31,32)$}
        \label{fig:sub1}
    \end{subfigure}
    \hfill
    \begin{subfigure}[b]{0.3\textwidth}
        \centering
        \fbox{%
        \includegraphics[width=\textwidth]{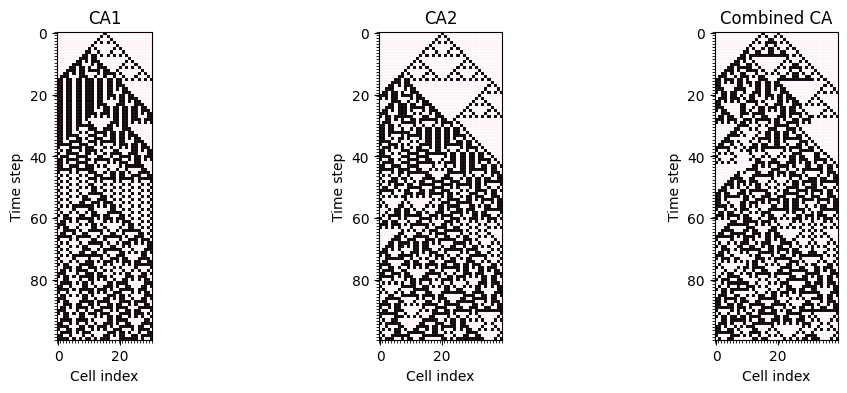}
        }
        \caption{$(k_1,k_2):(31,40)$}
        \label{fig:sub2}
    \end{subfigure}
    \hfill
    \begin{subfigure}[b]{0.3\textwidth}
    \centering
    \fbox{%
    \includegraphics[width=\textwidth]{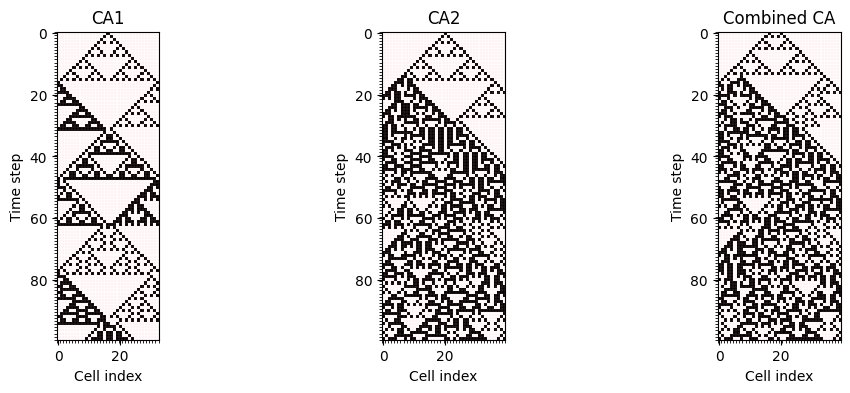}
    }
    \caption{$(k_1,k_2):(33,40)$}
    \label{fig:left}
  \end{subfigure}
  \hfill
   \begin{subfigure}[b]{0.3\textwidth}
    \centering
    \fbox{%
    \includegraphics[width=\textwidth]{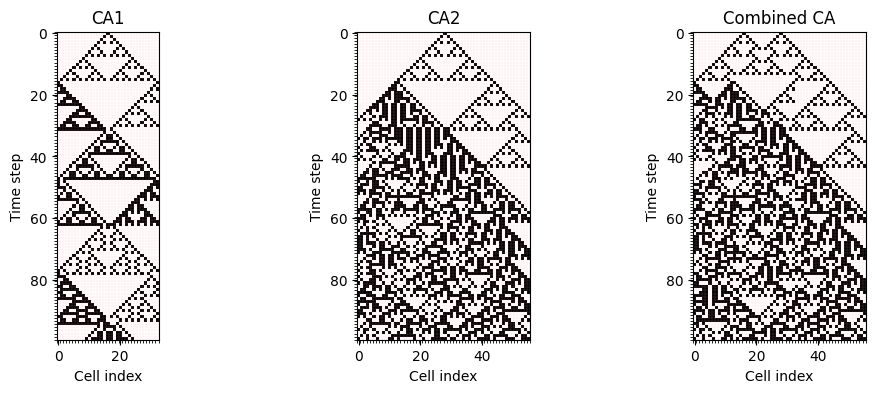}
    }
    \caption{$(k_1,k_2):(33,56)$}
    \label{fig:left}
  \end{subfigure}
  \hfill
     \begin{subfigure}[b]{0.3\textwidth}
    \centering
    \fbox{%
    \includegraphics[width=\textwidth]{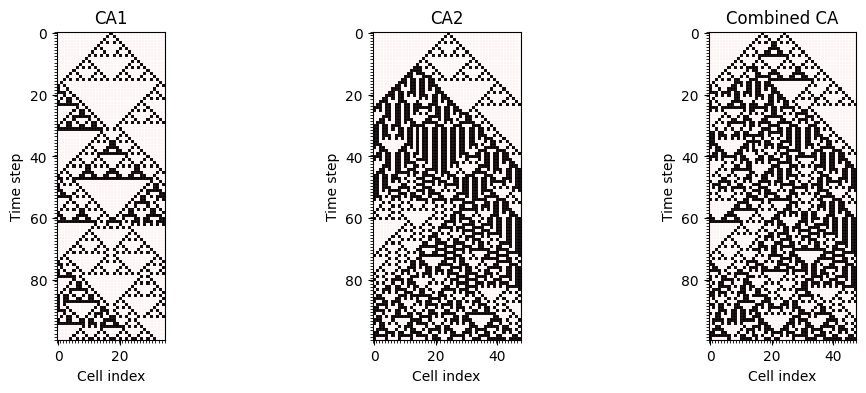}
    }
    \caption{$(k_1,k_2):(35,48)$}
    \label{fig:left}
  \end{subfigure}
  \hfill
   \begin{subfigure}[b]{0.3\textwidth}
    \centering
    \fbox{%
    \includegraphics[width=\textwidth]{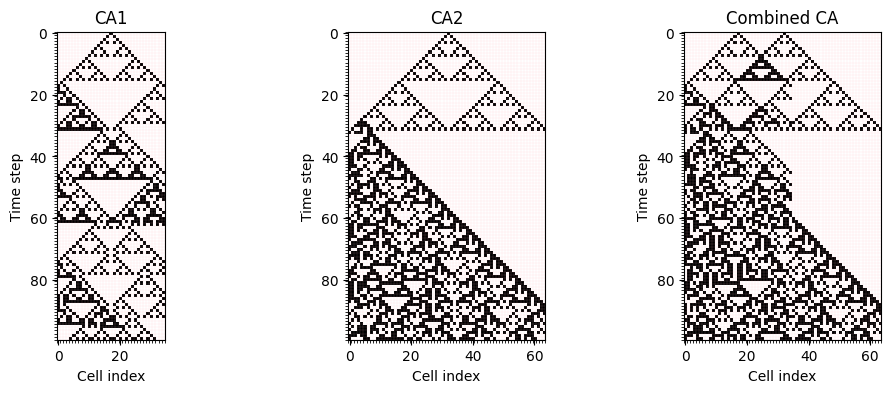}
    }
    \caption{$(k_1,k_2):(35,64)$}
    \label{fig:left}
  \end{subfigure}
  \hfill
  \begin{subfigure}[b]{0.3\textwidth}
    \centering
    \fbox{%
    \includegraphics[width=\textwidth]{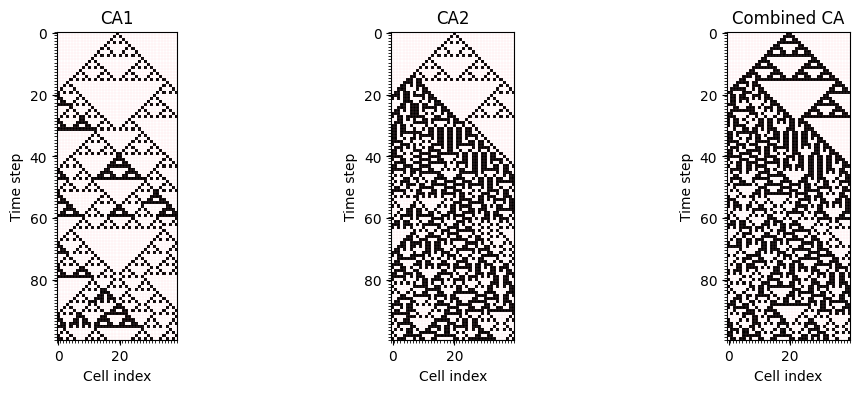}
    }
    \caption{$(k_1,k_2):(39,40)$}
    \label{fig:left}
  \end{subfigure}
  \hfill
  \begin{subfigure}[b]{0.3\textwidth}
    \centering
    \fbox{%
    \includegraphics[width=\textwidth]{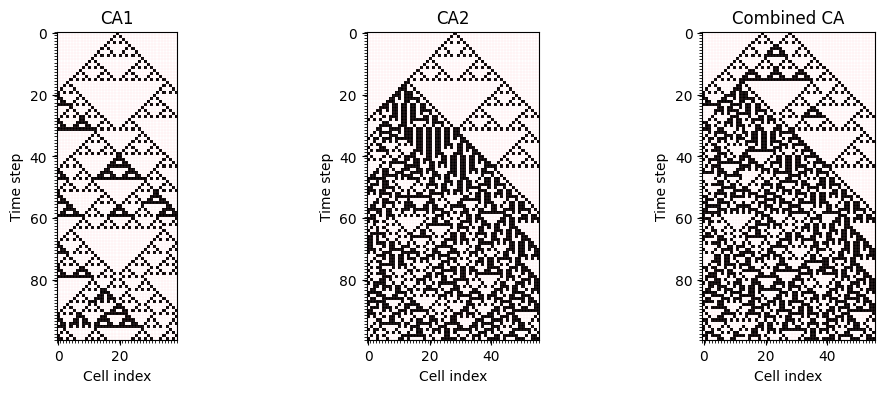}
    }
    \caption{$(k_1,k_2):(39,56)$}
    \label{fig:left}
  \end{subfigure}
  \hfill
  \begin{subfigure}[b]{0.3\textwidth}
    \centering
    \fbox{%
    \includegraphics[width=\textwidth]{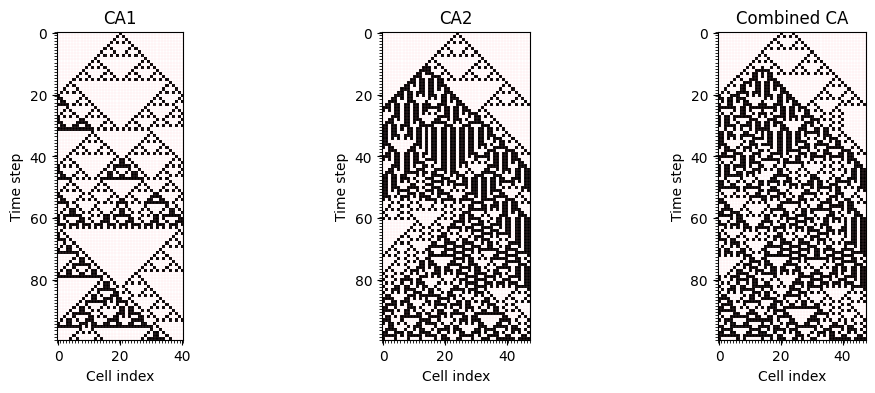}
    }
    \caption{$(k_1,k_2):(41,48)$}
    \label{fig:right}
  \end{subfigure}
   \hfill
    \begin{subfigure}[b]{0.3\textwidth}
    \centering
    \fbox{%
    \includegraphics[width=\textwidth]{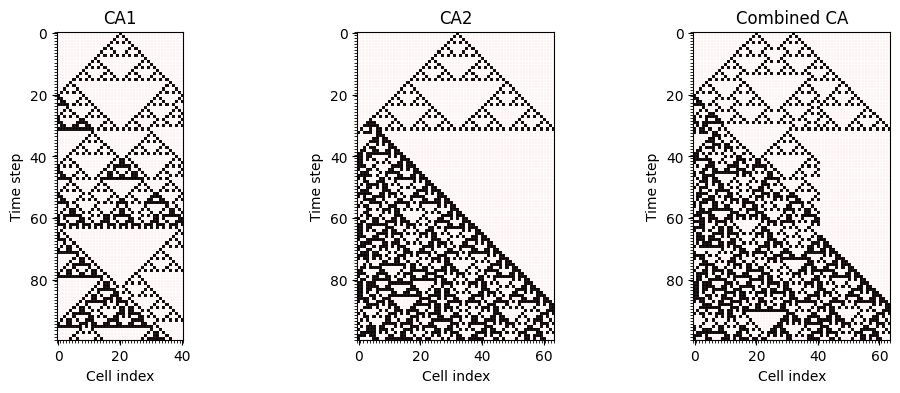}
    }
    \caption{$(k_1,k_2):(41,64)$}
    \label{fig:left}
  \end{subfigure}
  \hfill
   \begin{subfigure}[b]{0.3\textwidth}
    \centering
    \fbox{%
    \includegraphics[width=\textwidth]{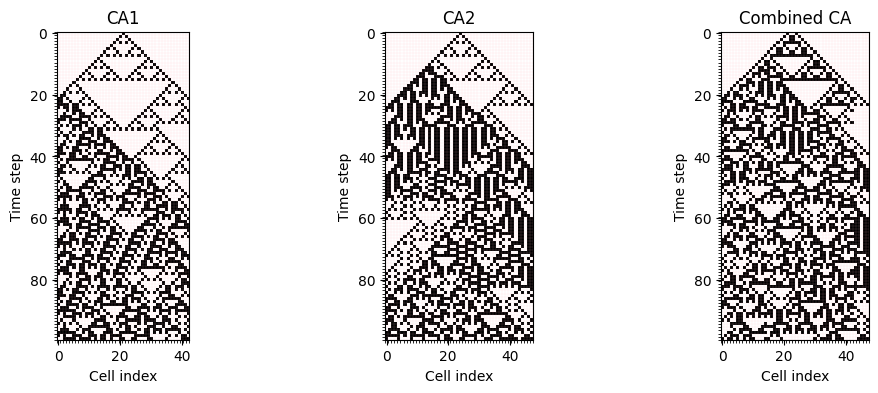}
    }
    \caption{$(k_1,k_2):(43,48)$}
    \label{fig:left}
  \end{subfigure}
  \hfill
   \begin{subfigure}[b]{0.3\textwidth}
    \centering
    \fbox{%
    \includegraphics[width=\textwidth]{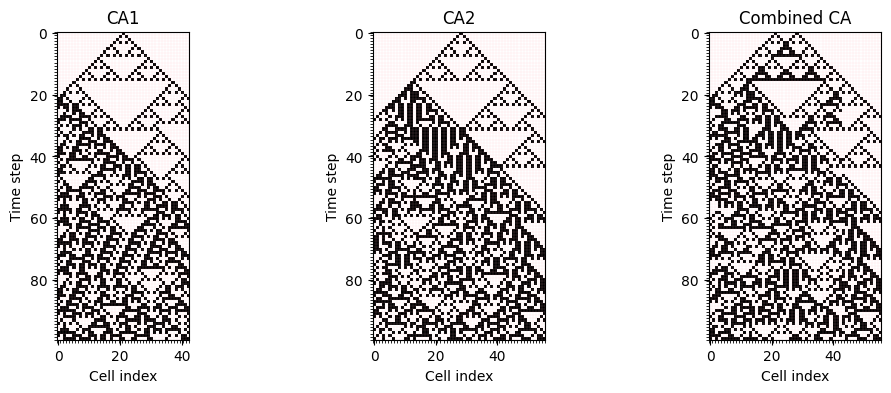}
    }
    \caption{$(k_1,k_2):(43,56)$}
    \label{fig:left}
  \end{subfigure}
  \hfill
   \begin{subfigure}[b]{0.3\textwidth}
    \centering
    \fbox{%
    \includegraphics[width=\textwidth]{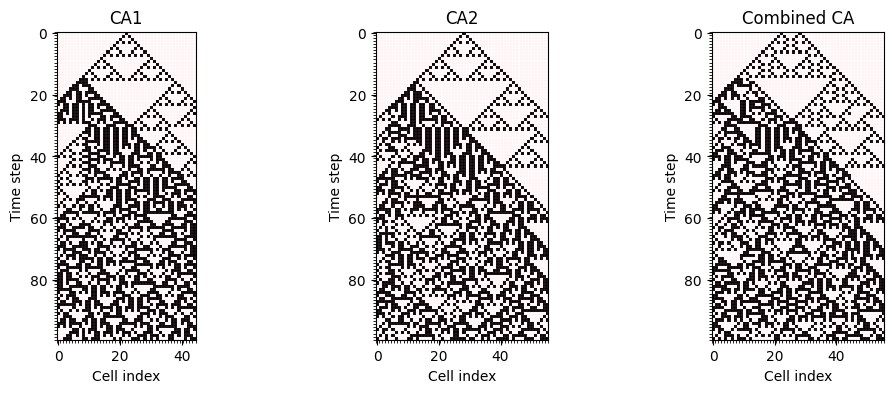}
    }
    \caption{$(k_1,k_2):(45,56)$}
    \label{fig:left}
  \end{subfigure}
  \hfill
  \begin{subfigure}[b]{0.3\textwidth}
    \centering
    \fbox{%
    \includegraphics[width=\textwidth]{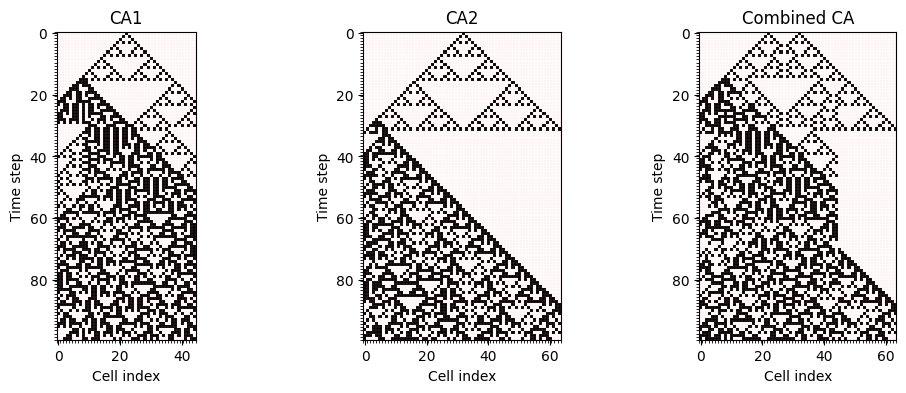}
    }
    \caption{$(k_1,k_2):(45,64)$}
    \label{fig:left}
  \end{subfigure}
   \hfill
  \begin{subfigure}[b]{0.3\textwidth}
    \centering
    \fbox{%
    \includegraphics[width=\textwidth]{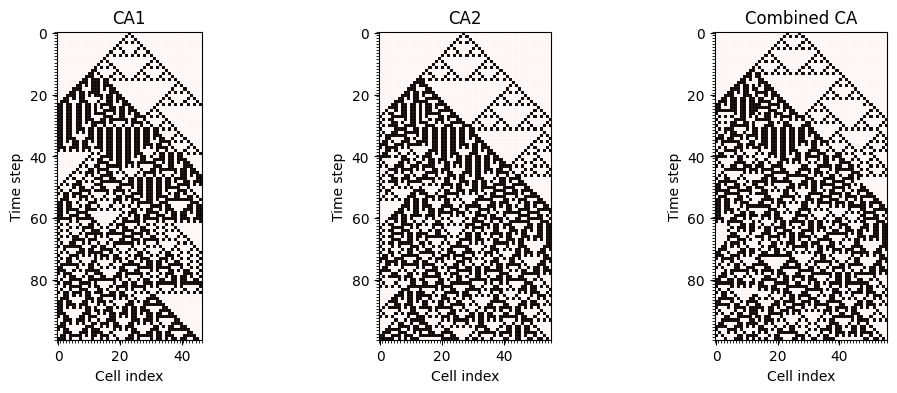}
    }
    \caption{$(k_1,k_2):(47,56)$}
    \label{fig:right}
  \end{subfigure}
   \hfill
   \begin{subfigure}[b]{0.3\textwidth}
    \centering
    \fbox{%
    \includegraphics[width=\textwidth]{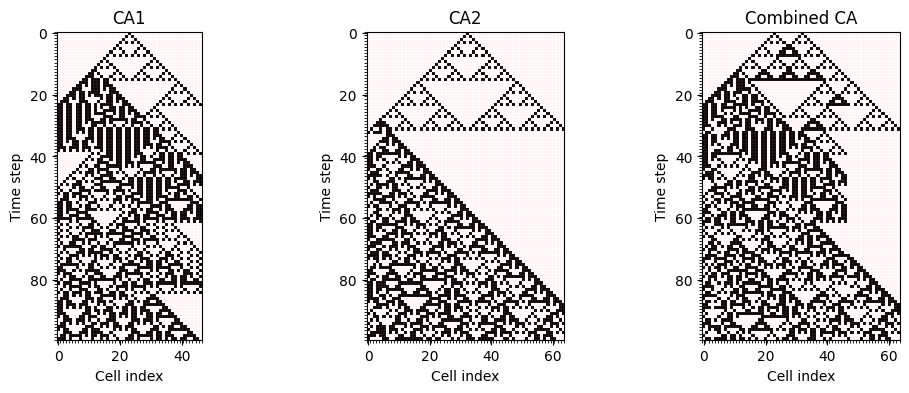}
    }
    \caption{$(k_1,k_2):(47,64)$}
    \label{fig:left}
  \end{subfigure}
  \hfill
  \begin{subfigure}[b]{0.3\textwidth}
    \centering
    \fbox{%
    \includegraphics[width=\textwidth]{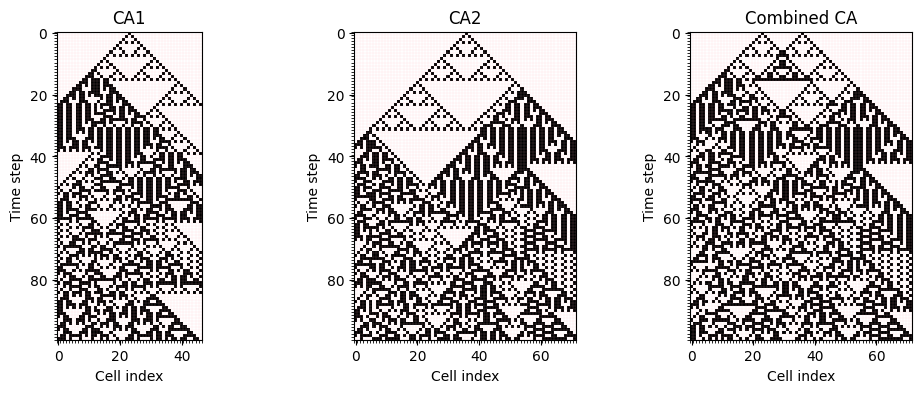}
    }
    \caption{$(k_1,k_2):(47,72)$}
    \label{fig:left}
  \end{subfigure}
  \hfill
  \begin{subfigure}[b]{0.3\textwidth}
    \centering
    \fbox{%
    \includegraphics[width=\textwidth]{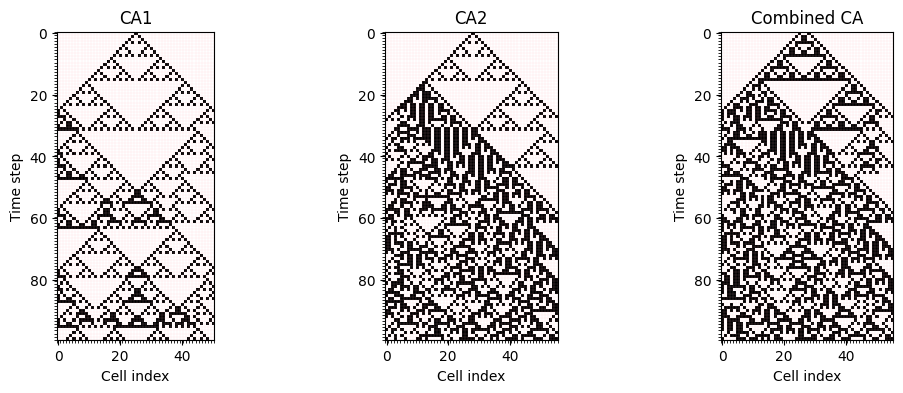}
    }
    \caption{$(k_1,k_2):(51,56)$}
    \label{fig:left}
  \end{subfigure}
  \hfill
  \begin{subfigure}[b]{0.3\textwidth}
    \centering
    \fbox{%
    \includegraphics[width=\textwidth]{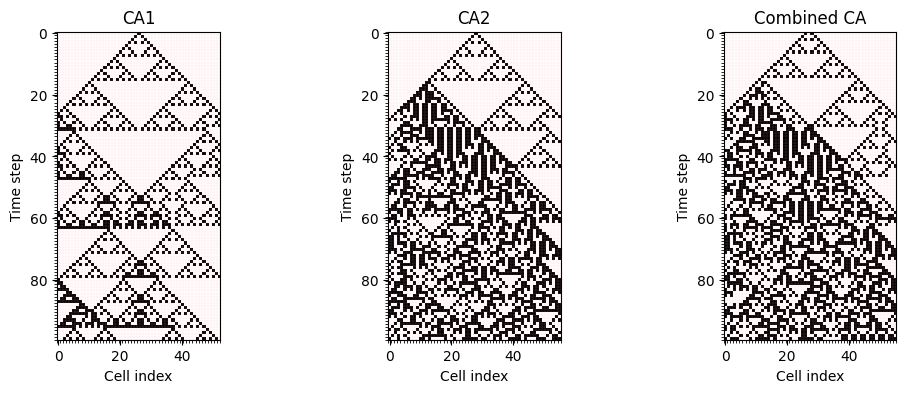}
    }
    \caption{$(k_1,k_2):(53,56)$}
    \label{fig:left}
  \end{subfigure}
  \hfill
  \begin{subfigure}[b]{0.3\textwidth}
    \centering
    \fbox{%
    \includegraphics[width=\textwidth]{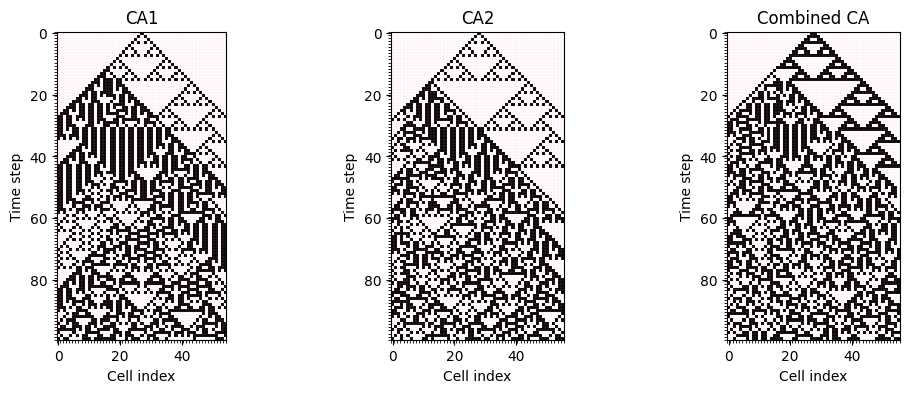}
    }
    \caption{$(k_1,k_2):(55,56)$}
    \label{fig:left}
  \end{subfigure}
   \hfill
   \begin{subfigure}[b]{0.3\textwidth}
    \centering
    \fbox{%
    \includegraphics[width=\textwidth]{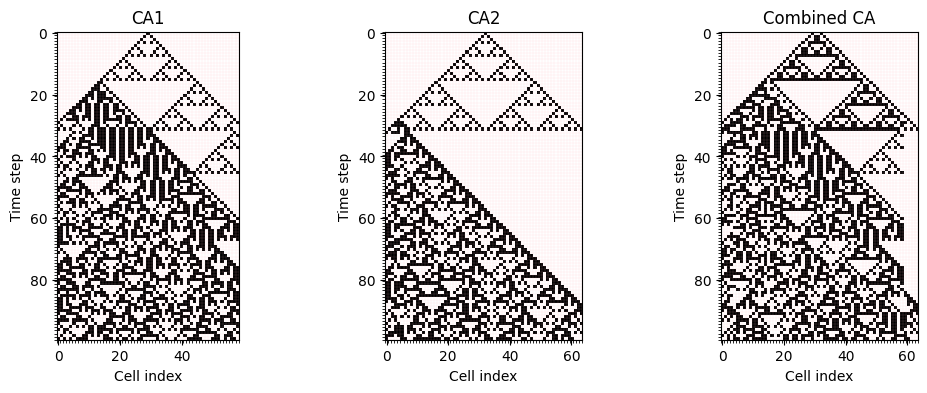}
    }
    \caption{$(k_1,k_2):(59,64)$}
    \label{fig:left}
  \end{subfigure}
  \hfill
  \begin{subfigure}[b]{0.3\textwidth}
    \centering
    \fbox{%
    \includegraphics[width=\textwidth]{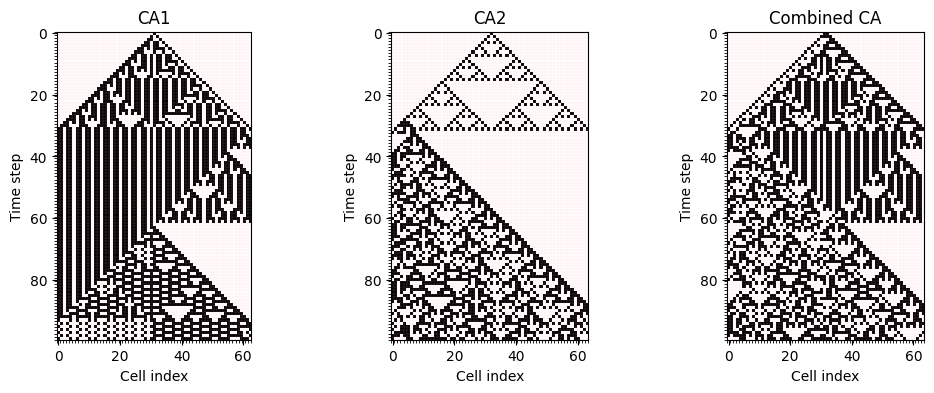}
    }
    \caption{$(k_1,k_2):(63,64)$}
    \label{fig:right}
  \end{subfigure}
  \hfill
   \begin{subfigure}[b]{0.3\textwidth}
    \centering
    \fbox{%
    \includegraphics[width=\textwidth]{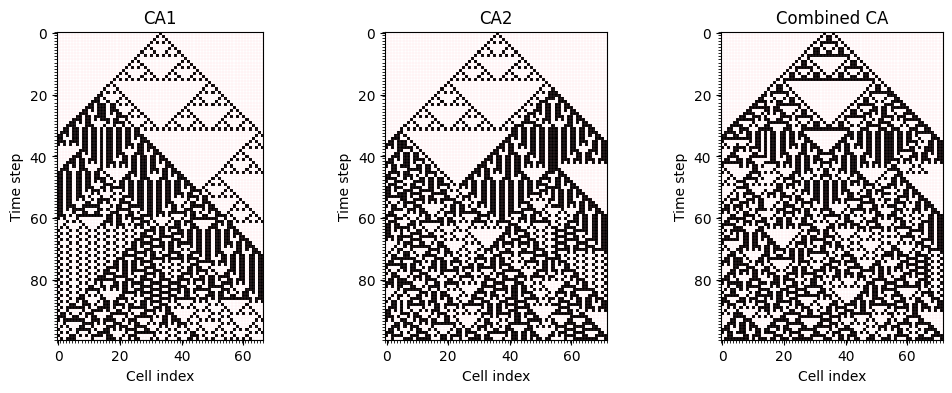}
    }
    \caption{$(k_1,k_2):(67,72)$}
    \label{fig:left}
  \end{subfigure}
  \hfill
   \begin{subfigure}[b]{0.3\textwidth}
    \centering
    \fbox{%
    \includegraphics[width=\textwidth]{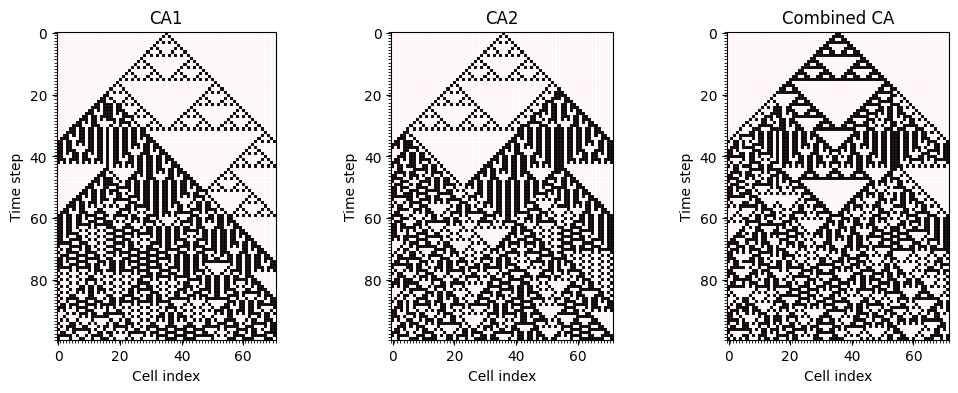}
    }
    \caption{$(k_1,k_2):(71,72)$}
    \label{fig:left}
  \end{subfigure}
  \caption{Components of Combined CA-based PRNGs with $s=1$ from \autoref{table14}}
  \label{fig:space-time1}
\end{figure}

\begin{figure}[!hbtp]
\vspace{-4.5em}
  \centering
  \begin{subfigure}[b]{0.3\textwidth}
    \centering
    \fbox{%
    \includegraphics[width=\textwidth]{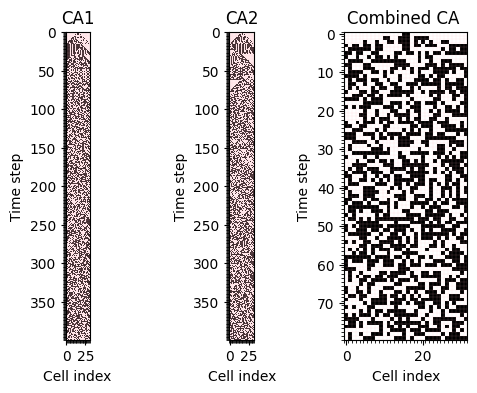}
    }
    \caption{$(k_1,k_2,s):(31,32,5)$}
    \label{fig:left}
  \end{subfigure}
  \hfill
  \begin{subfigure}[b]{0.3\textwidth}
    \centering
    \fbox{%
    \includegraphics[width=\textwidth]{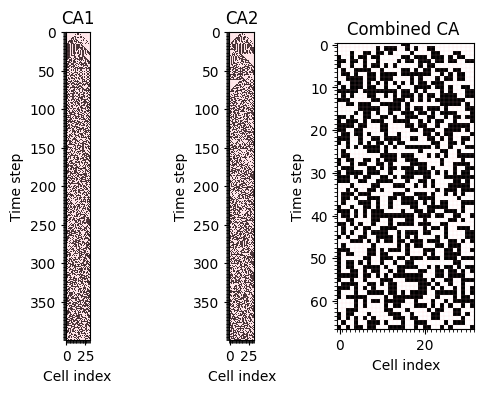}
    }
    \caption{$(k_1,k_2,s):(31,32,6)$}
    \label{fig:right}
  \end{subfigure}
  \hfill
  \begin{subfigure}[b]{0.3\textwidth}
    \centering
    \fbox{%
    \includegraphics[width=\textwidth]{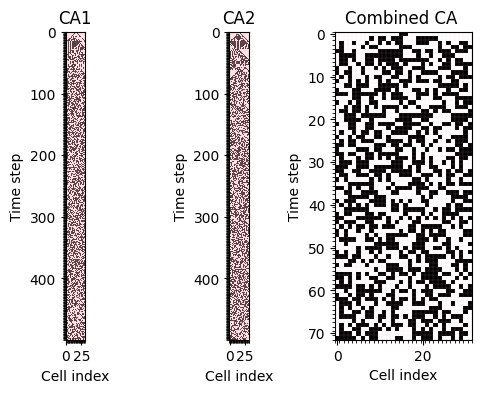}
    }
    \caption{$(k_1,k_2,s):(31,32,7)$}
    \label{fig:left}
  \end{subfigure}
  \hfill
  \begin{subfigure}[b]{0.3\textwidth}
    \centering
    \fbox{%
    \includegraphics[width=\textwidth]{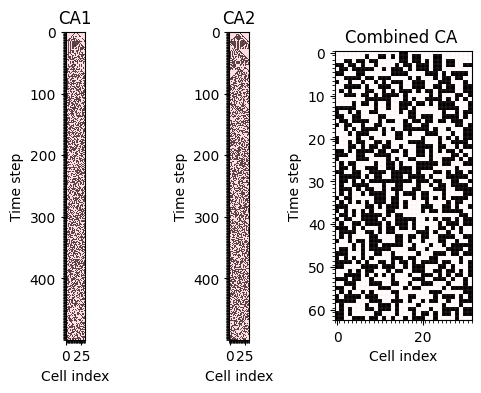}
    }
    \caption{$(k_1,k_2,s):(31,32,8)$}
    \label{fig:right}
  \end{subfigure}
  \hfill
  \begin{subfigure}[b]{0.3\textwidth}
    \centering
    \fbox{%
    \includegraphics[width=\textwidth]{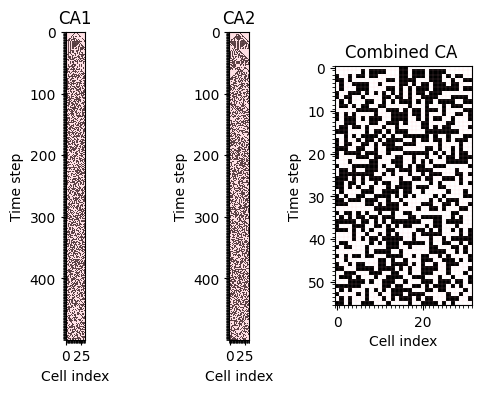}
    }
    \caption{$(k_1,k_2,s):(31,32,9)$}
    \label{fig:left}
  \end{subfigure}
  \hfill
  \begin{subfigure}[b]{0.3\textwidth}
    \centering
    \fbox{%
    \includegraphics[width=\textwidth]{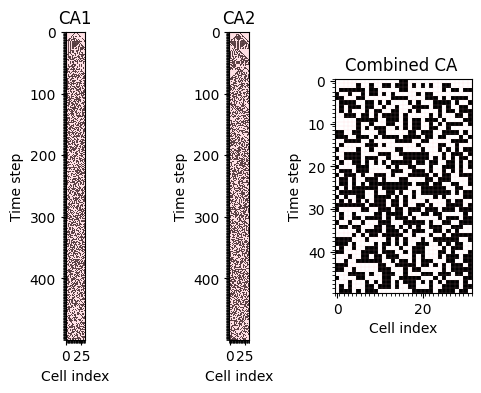}
    }
    \caption{$(k_1,k_2,s):(31,32,10)$}
    \label{fig:right}
  \end{subfigure}
  \hfill
  \begin{subfigure}[b]{0.3\textwidth}
    \centering
    \fbox{%
    \includegraphics[width=\textwidth]{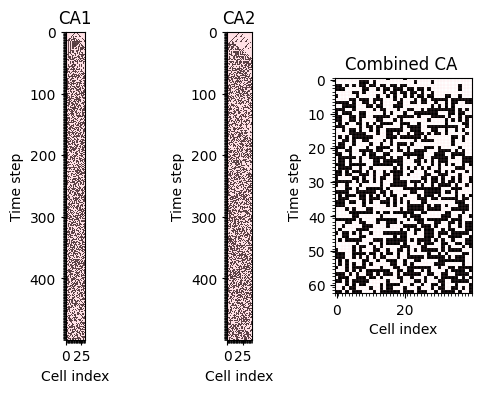}
    }
    \caption{$(k_1,k_2,s):(31,40,8)$}
    \label{fig:left}
  \end{subfigure}
  \hfill
  \begin{subfigure}[b]{0.3\textwidth}
    \centering
    \fbox{%
    \includegraphics[width=\textwidth]{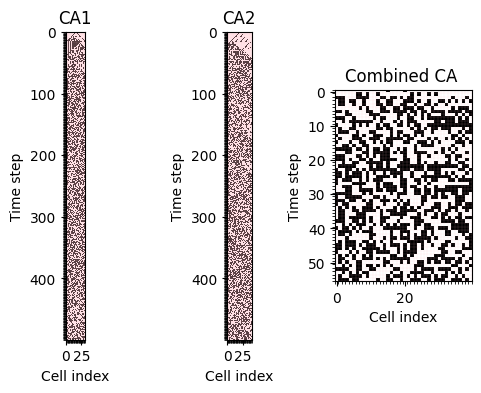}
    }
    \caption{$(k_1,k_2,s):(31,40,9)$}
    \label{fig:right}
  \end{subfigure}
   \hfill
  \begin{subfigure}[b]{0.3\textwidth}
    \centering
    \fbox{%
    \includegraphics[width=\textwidth]{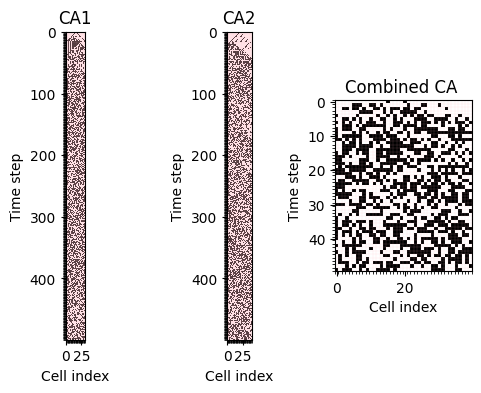}
    }
    \caption{$(k_1,k_2,s):(31,40,10)$}
    \label{fig:left}
  \end{subfigure}
  \hfill
  \begin{subfigure}[b]{0.3\textwidth}
    \centering
    \fbox{%
    \includegraphics[width=\textwidth]{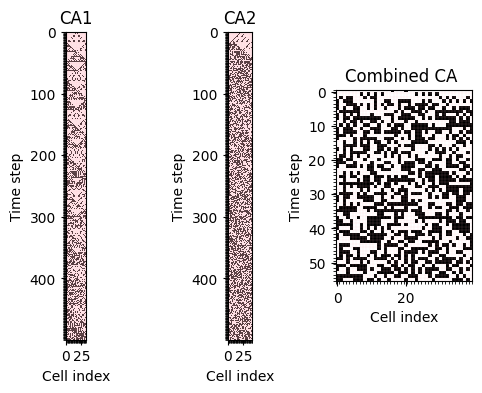}
    }
    \caption{$(k_1,k_2,s):(33,40,9)$}
    \label{fig:left}
  \end{subfigure}
  \hfill
  \begin{subfigure}[b]{0.3\textwidth}
    \centering
    \fbox{%
    \includegraphics[width=\textwidth]{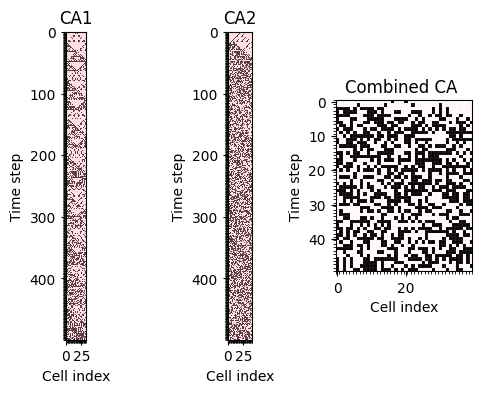}
    }
    \caption{$(k_1,k_2,s):(33,40,10)$}
    \label{fig:left}
  \end{subfigure}
  \hfill
  \begin{subfigure}[b]{0.3\textwidth}
    \centering
    \fbox{%
    \includegraphics[width=\textwidth]{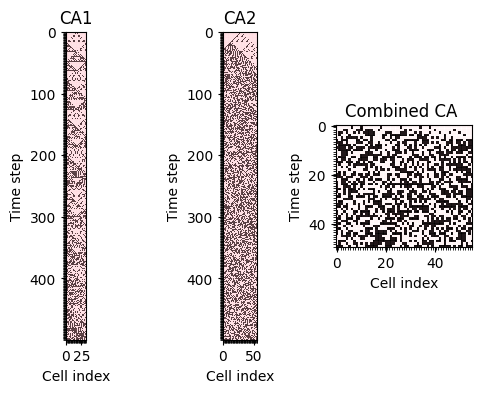}
    }
    \caption{$(k_1,k_2,s):(33,56,10)$}
    \label{fig:left}
  \end{subfigure}
  \begin{subfigure}[b]{0.3\textwidth}
    \centering
    \fbox{%
    \includegraphics[width=\textwidth]{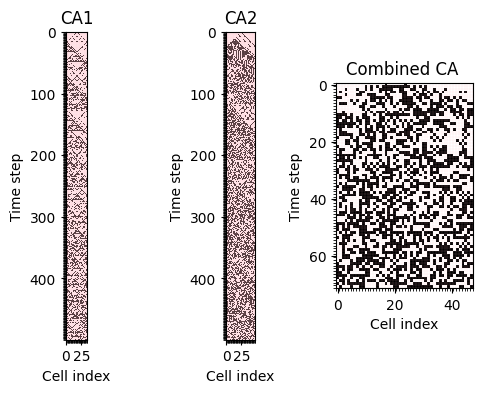}
    }
    \caption{$(k_1,k_2,s):(35,48,7)$}
    \label{fig:right}
  \end{subfigure}
  \hfill
  \begin{subfigure}[b]{0.3\textwidth}
    \centering
    \fbox{%
    \includegraphics[width=\textwidth]{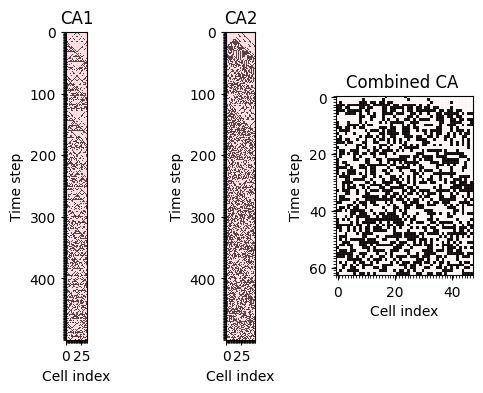}
    }
    \caption{$(k_1,k_2,s):(35,48,8)$}
    \label{fig:left}
  \end{subfigure}
  \hfill
  \begin{subfigure}[b]{0.3\textwidth}
    \centering
    \fbox{%
    \includegraphics[width=\textwidth]{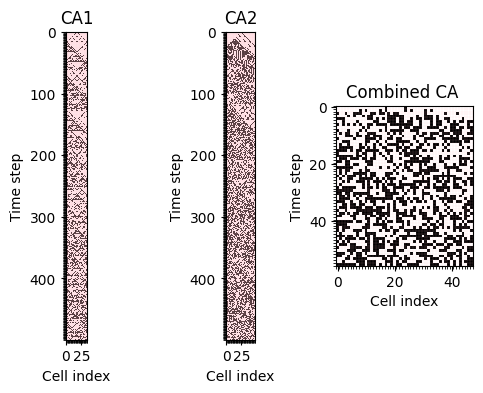}
    }
    \caption{$(k_1,k_2,s):(35,48,9)$}
    \label{fig:right}
  \end{subfigure}
  \caption{Space-time diagram for Combined CA-based PRNGs with time spacing}
  \label{fig:st1}
  \vspace{-1.5em}
\end{figure}

\begin{figure}[!hbtp]
\vspace{-4.5em}
  \centering
  \begin{subfigure}[b]{0.3\textwidth}
    \centering
    \fbox{%
    \includegraphics[width=\textwidth]{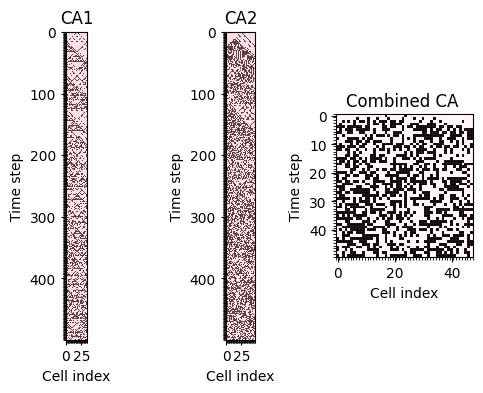}
    }
    \caption{$(k_1,k_2,s):(35,48,10)$}
    \label{fig:left}
  \end{subfigure}
  \hfill
  \begin{subfigure}[b]{0.3\textwidth}
    \centering
    \fbox{%
    \includegraphics[width=\textwidth]{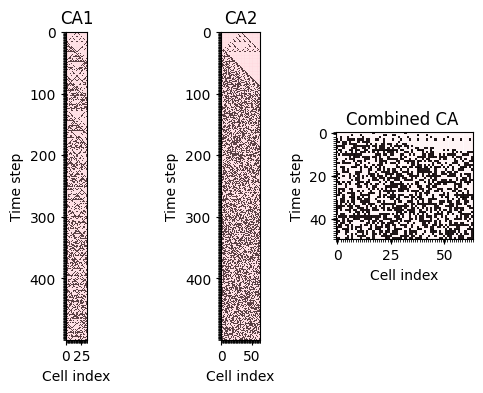}
    }
    \caption{$(k_1,k_2,s):(35,64,10)$}
    \label{fig:left}
  \end{subfigure}
  \hfill
  \begin{subfigure}[b]{0.3\textwidth}
    \centering
    \fbox{%
    \includegraphics[width=\textwidth]{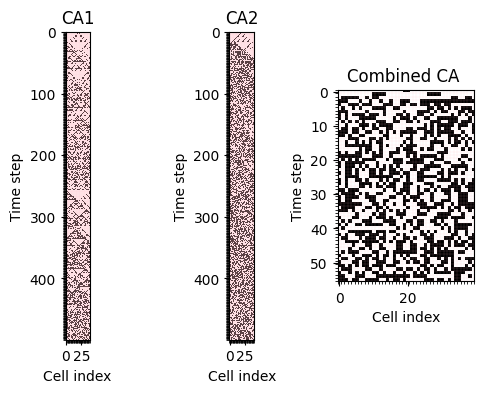}
    }
    \caption{$(k_1,k_2,s):(39,40,9)$}
    \label{fig:left}
  \end{subfigure}
  \hfill
  \begin{subfigure}[b]{0.3\textwidth}
    \centering
    \fbox{%
    \includegraphics[width=\textwidth]{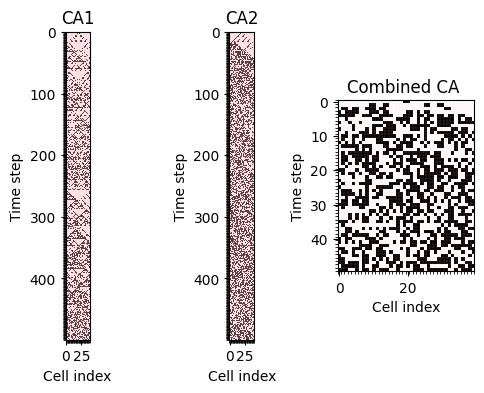}
    }
    \caption{$(k_1,k_2,s):(39,40,10)$}
    \label{fig:left}
  \end{subfigure}
  \hfill
  \begin{subfigure}[b]{0.3\textwidth}
    \centering
    \fbox{%
    \includegraphics[width=\textwidth]{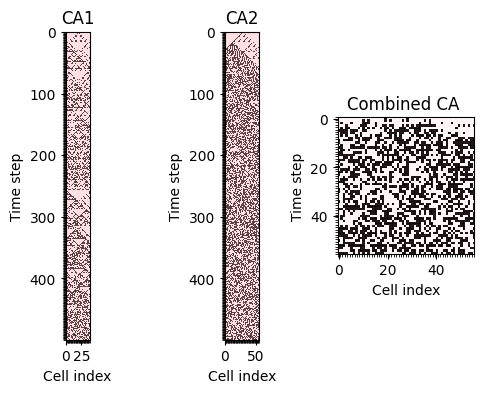}
    }
    \caption{$(k_1,k_2,s):(39,56,9)$}
    \label{fig:left}
  \end{subfigure}
  \hfill
  \begin{subfigure}[b]{0.3\textwidth}
    \centering
    \fbox{%
    \includegraphics[width=\textwidth]{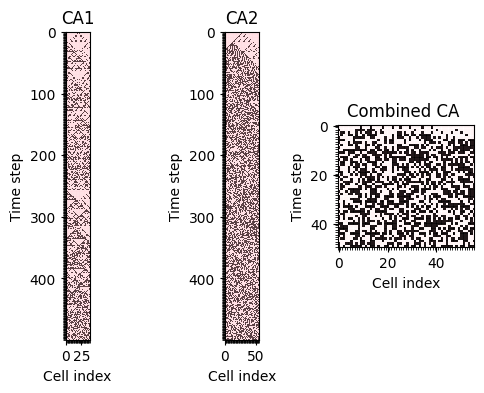}
    }
    \caption{$(k_1,k_2,s):(39,56,10)$}
    \label{fig:left}
  \end{subfigure}
  \hfill
  \begin{subfigure}[b]{0.3\textwidth}
    \centering
    \fbox{%
    \includegraphics[width=\textwidth]{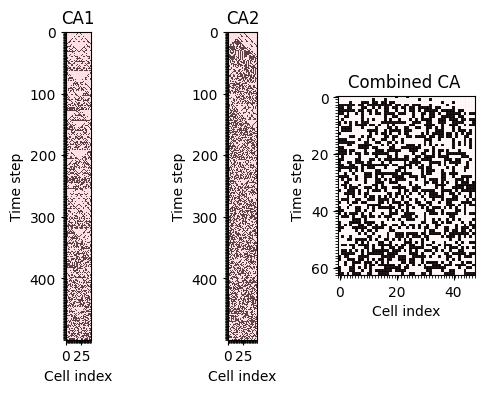}
    }
    \caption{$(k_1,k_2,s):(41,48,8)$}
    \label{fig:right}
  \end{subfigure}
   \hfill
  \begin{subfigure}[b]{0.3\textwidth}
    \centering
    \fbox{%
    \includegraphics[width=\textwidth]{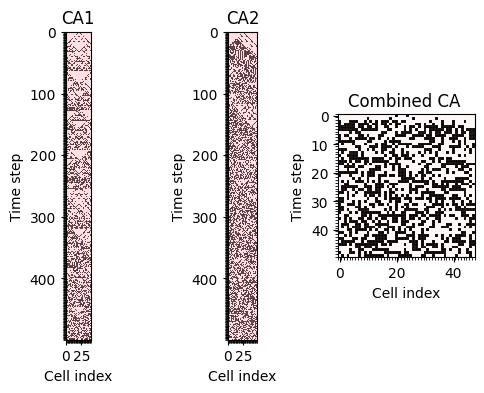}
    }
    \caption{$(k_1,k_2,s):(41,48,10)$}
    \label{fig:left}
  \end{subfigure}
  \hfill
  \begin{subfigure}[b]{0.3\textwidth}
    \centering
    \fbox{%
    \includegraphics[width=\textwidth]{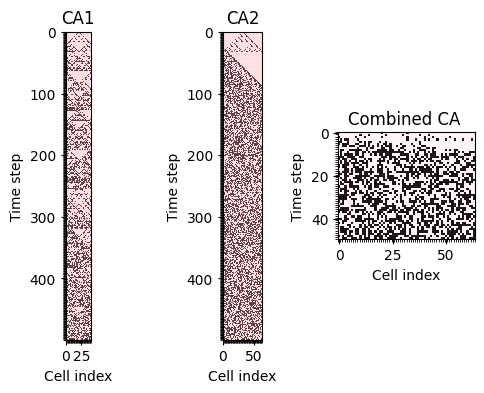}
    }
    \caption{$(k_1,k_2,s):(41,64,10)$}
    \label{fig:left}
  \end{subfigure}
  \hfill
  \begin{subfigure}[b]{0.3\textwidth}
    \centering
    \fbox{%
    \includegraphics[width=\textwidth]{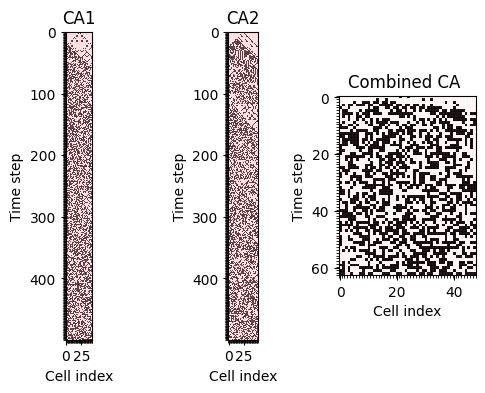}
    }
    \caption{$(k_1,k_2,s):(43,48,8)$}
    \label{fig:right}
  \end{subfigure}
  \hfill
  \begin{subfigure}[b]{0.3\textwidth}
    \centering
    \fbox{%
    \includegraphics[width=\textwidth]{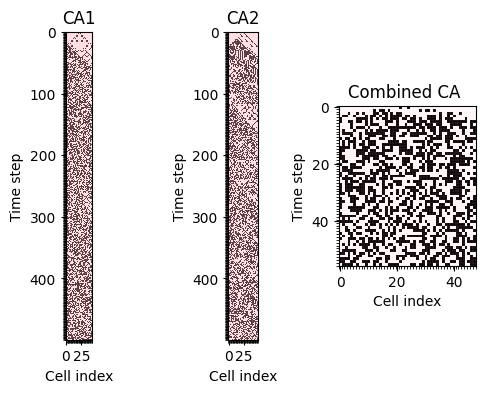}
    }
    \caption{$(k_1,k_2,s):(43,48,9)$}
    \label{fig:left}
  \end{subfigure}
  \hfill
  \begin{subfigure}[b]{0.3\textwidth}
    \centering
    \fbox{%
    \includegraphics[width=\textwidth]{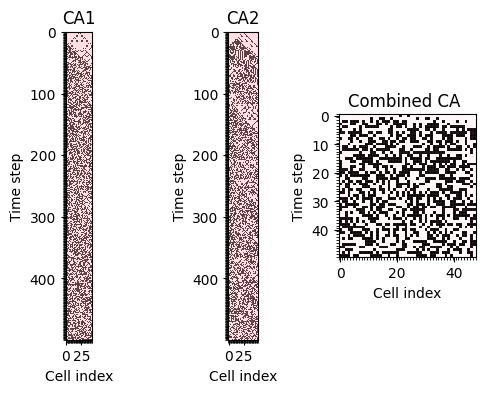}
    }
    \caption{$(k_1,k_2,s):(43,48,10)$}
    \label{fig:right}
  \end{subfigure}
    \hfill
    \begin{subfigure}[b]{0.3\textwidth}
    \centering
    \fbox{%
    \includegraphics[width=\textwidth]{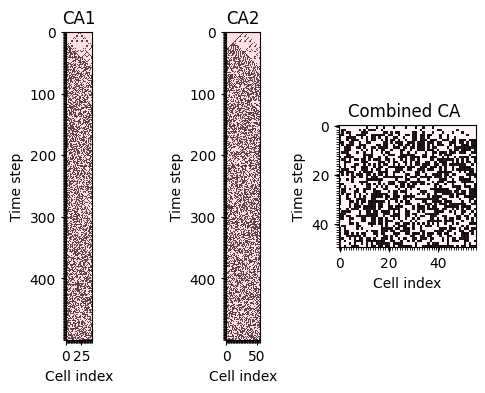}
    }
    \caption{$(k_1,k_2,s):(43,56,10)$}
    \label{fig:right}
  \end{subfigure}
  \hfill
    \begin{subfigure}[b]{0.3\textwidth}
    \centering
    \fbox{%
    \includegraphics[width=\textwidth]{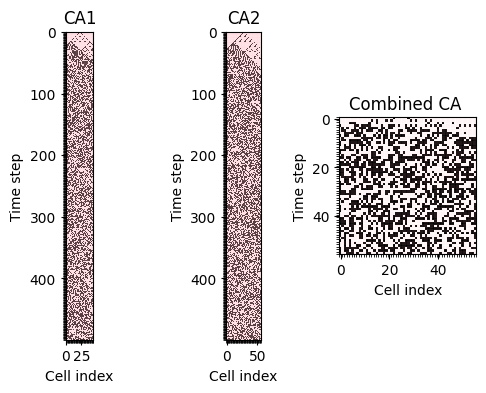}
    }
    \caption{$(k_1,k_2,s):(45,56,9)$}
    \label{fig:right}
  \end{subfigure}
  \hfill
    \begin{subfigure}[b]{0.3\textwidth}
    \centering
    \fbox{%
    \includegraphics[width=\textwidth]{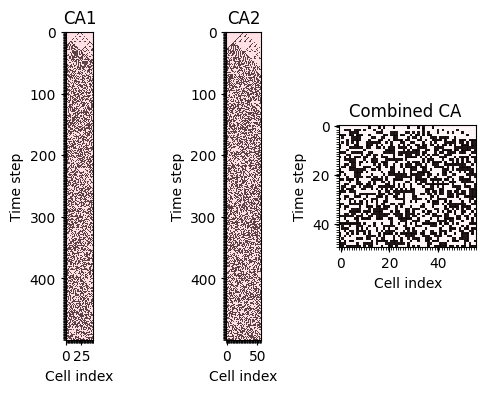}
    }
    \caption{$(k_1,k_2,s):(45,56,10)$}
    \label{fig:right}
  \end{subfigure}
\caption{Space-time diagram for Combined CA-based PRNGs with time spacing}
  \label{fig:st2}
  \vspace{-1.5em}
\end{figure}
\subsection{Verification with Space-time diagram}

Now, we explore the reasons why the combined generators are performing well by using space-time diagrams. The space-time diagram represents the graphical representation of the evolution of CA at each time step. In this space-time diagram, the x-axis denotes the CA configuration (or, the state of the PRNG), and the y-axis denotes the time steps. 
Here, black cells represent state 1, and white represents state 0. So using these space-time diagrams, we can easily analyze the behavior of CA based on the patterns generated from these diagrams \cite{PRNG1}.

We check the space-time diagrams for all combined maximal length CA-based PRNGs with time spacing from \autoref{table14} considering the initial configuration (seed) as a standard non-random pattern of middle bit as 1, and all other bits set to zero. First, we create the space-time diagram for the individual components and the combined generator without any time spacing. These are shown in Figure~\ref{fig:space-time1}. Here we observe that, the self-similar patterns of the component CAs are clearly visible in the combined CA as well when time spacing is not used. It can also be noted that, since XOR is used to combine, the inherent chaotic properties of the component CAs are not destroyed in the combined CA. Then, we create the space-time diagram with time spacing for all the combined PRNGs of \autoref{table14} and compare the behavior of individual components and the combined CA. Here, for each figure, the component CAs are evolved for $500$ time steps and according to the value of $s$, the number of unique random numbers are printed in the combined CA. These are shown in Figure~\ref{fig:st1} to Figure~\ref{fig:st5}. 
From these diagrams, we observe that, for all the combinations which perform very well in statistical tests, the patterns completely vanish in the combined generators' space-time diagrams and the numbers appear noisy in color; that proves their capability to be good source of randomness.

  \begin{figure}[!htbp]
  \vspace{-4.5em}
  \centering
    \begin{subfigure}[b]{0.3\textwidth}
    \centering
    \fbox{%
    \includegraphics[width=\textwidth]{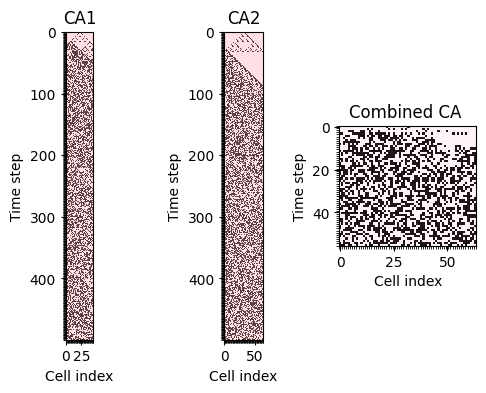}
    }
    \caption{$(k_1,k_2,s):(45,64,9)$}
    \label{fig:right}
  \end{subfigure}
  \hfill
    \begin{subfigure}[b]{0.3\textwidth}
    \centering
    \fbox{%
    \includegraphics[width=\textwidth]{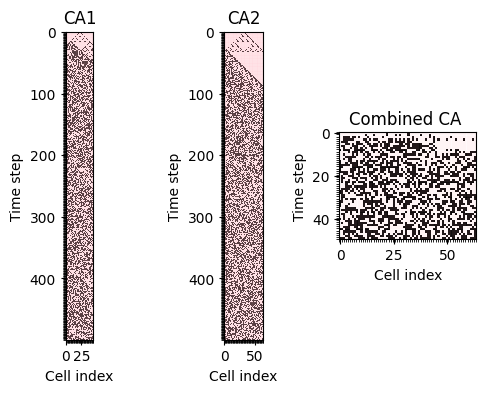}
    }
    \caption{$(k_1,k_2,s):(45,64,10)$}
    \label{fig:right}
  \end{subfigure}
  \hfill
  \begin{subfigure}[b]{0.3\textwidth}
    \centering
    \fbox{%
    \includegraphics[width=\textwidth]{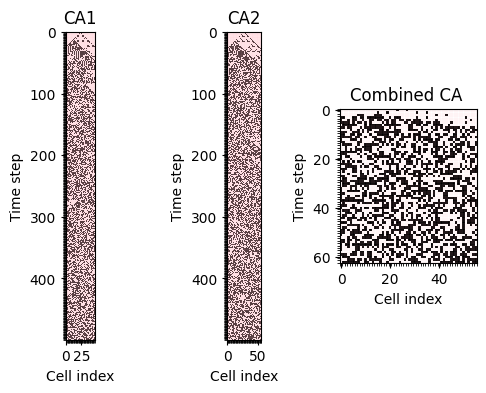}
    }
    \caption{$(k_1,k_2,s):(47,56,8)$}
    \label{fig:left}
  \end{subfigure}
  \hfill
  \begin{subfigure}[b]{0.3\textwidth}
    \centering
    \fbox{%
    \includegraphics[width=\textwidth]{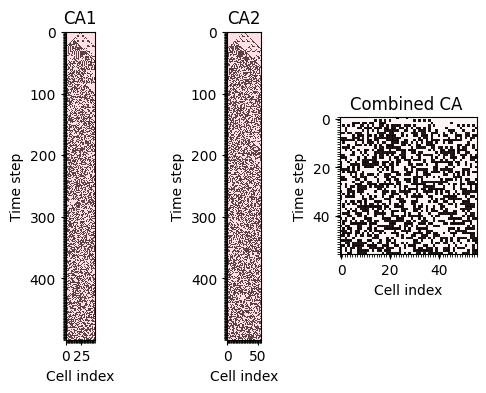}
    }
    \caption{$(k_1,k_2,s):(47,56,9)$}
    \label{fig:right}
  \end{subfigure}
\hfill
  \begin{subfigure}[b]{0.3\textwidth}
    \centering
    \fbox{%
    \includegraphics[width=\textwidth]{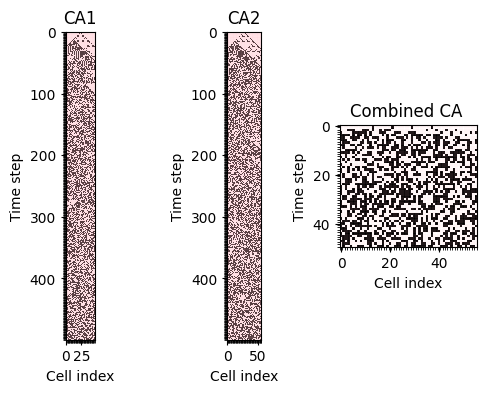}
    }
    \caption{$(k_1,k_2,s):(47,56,10)$}
    \label{fig:left}
  \end{subfigure}
  \hfill
    \begin{subfigure}[b]{0.3\textwidth}
    \centering
    \fbox{%
    \includegraphics[width=\textwidth]{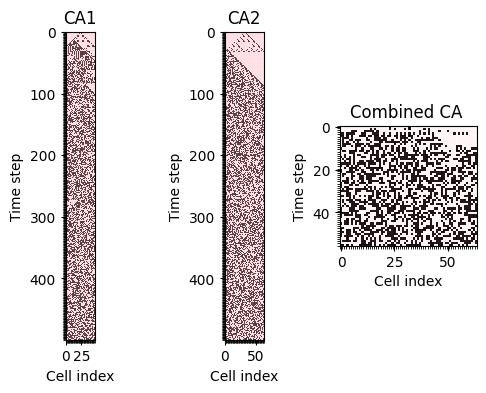}
    }
    \caption{$(k_1,k_2,s):(47,64,9)$}
    \label{fig:right}
  \end{subfigure}
  \hfill
    \begin{subfigure}[b]{0.3\textwidth}
    \centering
    \fbox{%
    \includegraphics[width=\textwidth]{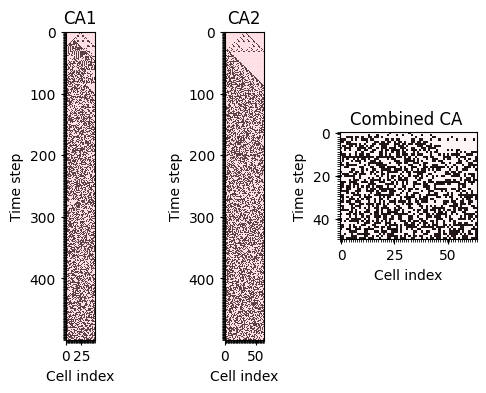}
    }
    \caption{$(k_1,k_2,s):(47,64,10)$}
    \label{fig:right}
  \end{subfigure}
  \hfill
    \begin{subfigure}[b]{0.3\textwidth}
    \centering
    \fbox{%
    \includegraphics[width=\textwidth]{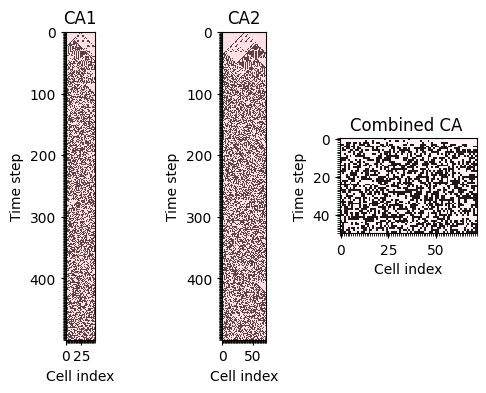}
    }
    \caption{$(k_1,k_2,s):(47,72,10)$}
    \label{fig:right}
  \end{subfigure}
  \hfill
    \begin{subfigure}[b]{0.3\textwidth}
    \centering
    \fbox{%
    \includegraphics[width=\textwidth]{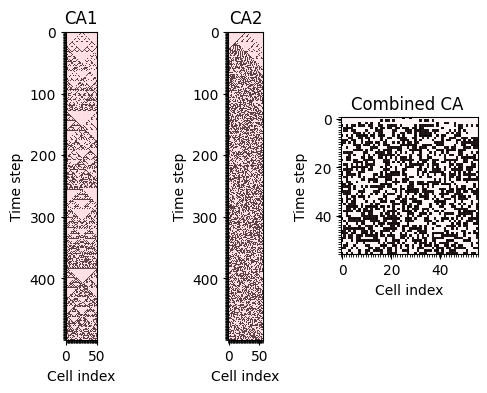}
    }
    \caption{$(k_1,k_2,s):(51,56,9)$}
    \label{fig:right}
  \end{subfigure}
    \begin{subfigure}[b]{0.3\textwidth}
    \centering
    \fbox{%
    \includegraphics[width=\textwidth]{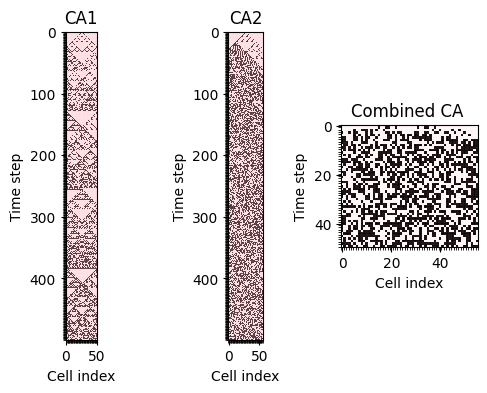}
    }
    \caption{$(k_1,k_2,s):(51,56,10)$}
    \label{fig:right}
    \end{subfigure}
    \hfill
    \begin{subfigure}[b]{0.3\textwidth}
    \centering
    \fbox{%
    \includegraphics[width=\textwidth]{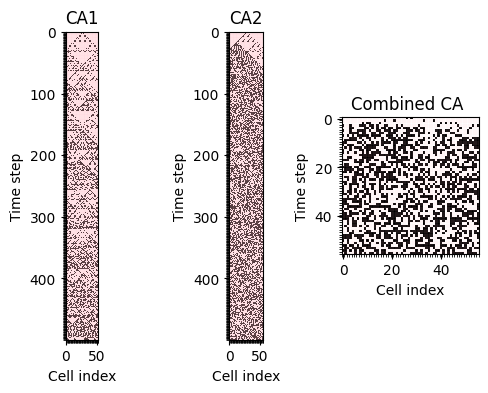}
    }
    \caption{$(k_1,k_2,s):(53,56,9)$}
    \label{fig:right}
  \end{subfigure}
  \hfill
    \begin{subfigure}[b]{0.3\textwidth}
    \centering
    \fbox{%
    \includegraphics[width=\textwidth]{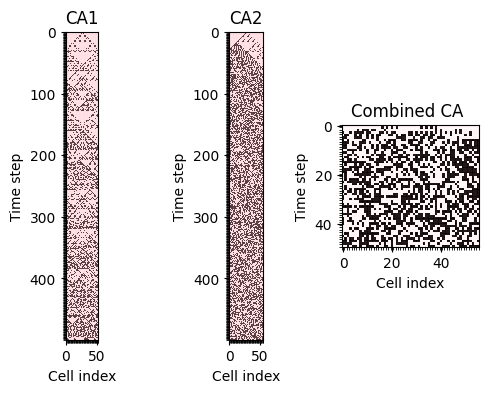}
    }
    \caption{$(k_1,k_2,s):(53,56,10)$}
    \label{fig:right}
  \end{subfigure}
   \hfill
    \begin{subfigure}[b]{0.3\textwidth}
    \centering
    \fbox{%
    \includegraphics[width=\textwidth]{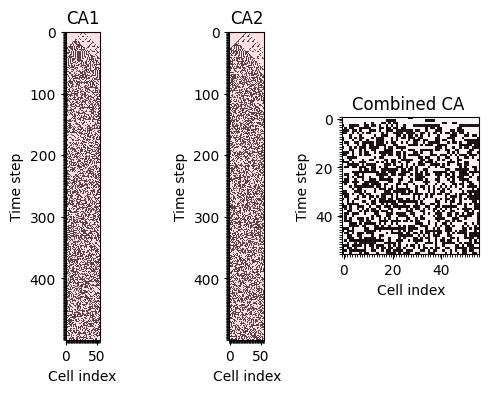}
    }
    \caption{$(k_1,k_2,s):(55,56,9)$}
    \label{fig:right}
  \end{subfigure}
   \hfill
    \begin{subfigure}[b]{0.3\textwidth}
    \centering
    \fbox{%
    \includegraphics[width=\textwidth]{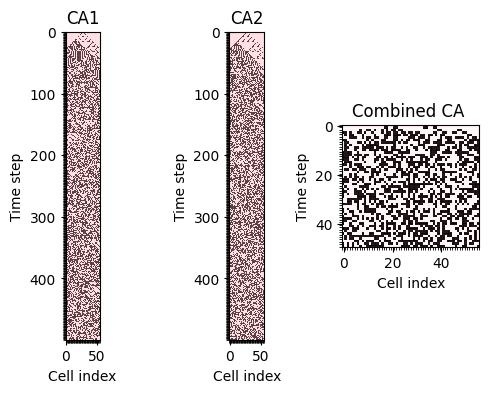}
    }
    \caption{$(k_1,k_2,s):(55,56,10)$}
    \label{fig:right}
  \end{subfigure}
  \hfill
  \begin{subfigure}[b]{0.3\textwidth}
    \centering
    \fbox{%
    \includegraphics[width=\textwidth]{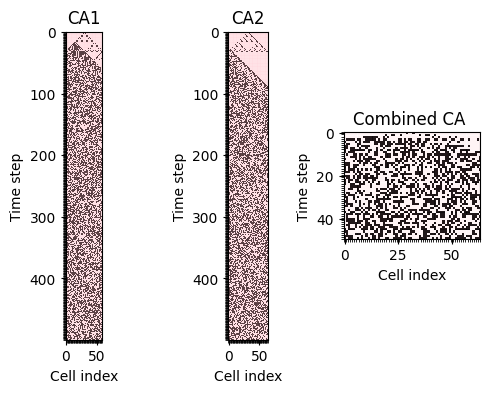}
    }
    \caption{$(k_1,k_2,s):(59,64,10)$}
    \label{fig:right}
  \end{subfigure}
  \hfill
 \caption{Space-time diagram for Combined CA-based PRNGs with time spacing}
  \label{fig:sp4}
  \vspace{-1.5em}
\end{figure}

  \begin{figure}[!htbp]
  \vspace{-2.5em}
  \centering
\begin{subfigure}[b]{0.3\textwidth}
    \centering
    \fbox{%
    \includegraphics[width=\textwidth]{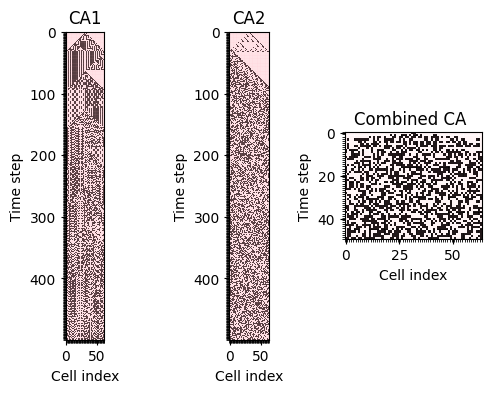}
    }
    \caption{$(k_1,k_2,s):(63,64,10)$}
    \label{fig:left}
  \end{subfigure}
   \hfill
    \begin{subfigure}[b]{0.3\textwidth}
    \centering
    \fbox{%
    \includegraphics[width=\textwidth]{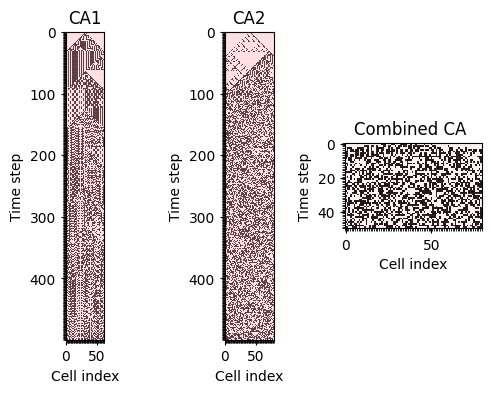}
    }
    \caption{$(k_1,k_2,s):(63,80,10)$}
    \label{fig:right}
  \end{subfigure}
  \hfill
  \begin{subfigure}[b]{0.3\textwidth}
    \centering
    \fbox{%
    \includegraphics[width=\textwidth]{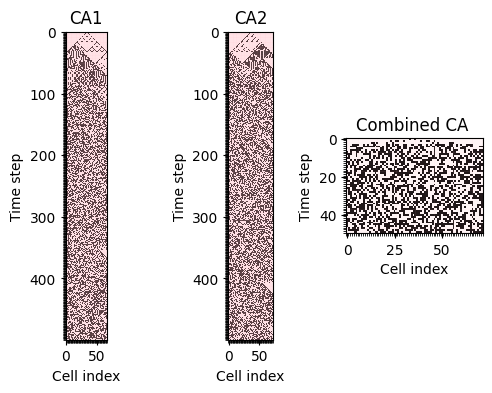}
    }
    \caption{$(k_1,k_2,s):(67,72,10)$}
    \label{fig:right}
  \end{subfigure}
   \hfill
    \begin{subfigure}[b]{0.3\textwidth}
    \centering
    \fbox{%
    \includegraphics[width=\textwidth]{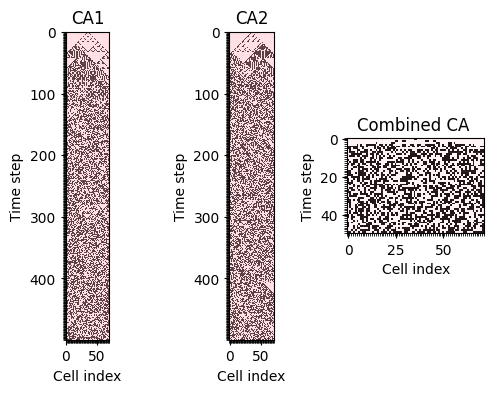}
    }
    \caption{$(k_1,k_2,s):(71,72,10)$}
    \label{fig:right}
  \end{subfigure}
   \caption{Space-time diagram for Combined CA-based PRNGs with time spacing}
  \label{fig:st5}
   \vspace{-1.0em}
\end{figure}

\subsection{Speed Test}
\label{section5.2}
 \autoref{figure2} depicts a sample code of our proposed combined CA-based PRNGs. This code implements the second algorithm described in \autoref{section5.2}. The generator is implemented and tested on both 32-bit and 64-bit computers. The performance evaluation is conducted by measuring the execution time required to generate $10^{9}$ pseudo-random numbers using the proposed combined CA-based PRNG, which satisfies the maximal period and maximal equidistribution with good statistical test results as discussed in \autoref{table14}. Then the speed of these combined CA-based PRNGs with existing linear generators such as Mersenne, GFSR4, WELL, and Tausworthe are compared as shown in \autoref{table13}. We observe that our PRNGs are faster than Mersenne Twister, but slower than the other PRNGs which do not use any time spacing.

\begin{figure}[!h]  
    \centering
    \fbox{%
        \includegraphics[width=0.7\textwidth]{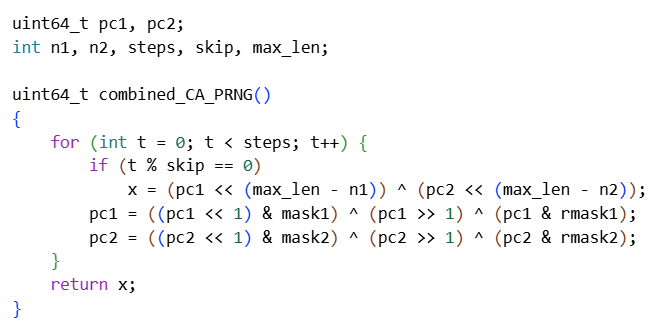}%
    }
\caption{Combined two-component CA-based PRNG }
\label{figure2}
\vspace{-1.0em}
\end{figure}

\begin{scriptsize}
\vspace{-1.0em}

\begin{longtable}{|l|c|c|}
\caption{Speed comparison of PRNGs} 
\label{table13} \\
\hline
\textbf{PRNG} & \textbf{Period Length} ($\boldsymbol{\rho}$) & \textbf{CPU Time (s)} \\
\hline
\endfirsthead
\hline
\textbf{PRNG} & \textbf{Period Length} ($\boldsymbol{\rho}$) & \textbf{CPU Time (s)} \\
\hline
\endhead

CA-PRNG $(31,32,7)$   & $2^{63}$  & 54 \\
CA-PRNG $(31,32,8)$   & $2^{63}$  & 60 \\
CA-PRNG $(31,40,8)$   & $2^{71}$  & 59   \\
CA-PRNG $(35,48,8)$   & $2^{83}$  & 87   \\
CA-PRNG $(41,48,8)$   & $2^{89}$  & 83 \\
CA-PRNG $(43,48,8)$   & $2^{91}$  & 84  \\
CA-PRNG $(47,56,8)$   & $2^{103}$ & 94   \\
Mersenne Twister       & $2^{19937}-1$ & 116 \\
Tausworthe (combined)  & $2^{88}$  & 76 \\
GFSR4                  & --        & 70 \\
WELL512a               & $2^{512}-1$ & 35 \\
WELL1024a              & $2^{1024}-1$ & 42 \\
\hline
\end{longtable}\vspace{-1.5em}

\end{scriptsize}

\section{Conclusion}
\label{section7}
In this work, we have proposed lightweight combined CA-based PRNGs using linear maximal length CAs that can satisfy the theoretical quality criteria of maximal equidistribution. We have taken two sources of maximal length CAs from Ref.~\cite{CA4,CA7} up to degree $128$ which use minimal number of different rules in the rule vector. The basic combined CA-based PRNGs without time spacing achieve close to maximal period length but not the maximally equidistributed characteristic. So, to improve the equidistribution, we apply time spacing over the combined CA-based PRNGs. But to keep them light-weight, we restrict the time spacing to upto 10 steps. These PRNGs achieve the period length close to maximal as well as the maximal equidistribution. Then experiments are carried out to assess the statistical quality of these combined CA-based PRNG with time spacing on benchmark statistical testbeds like Dieharder and TestU01. We identify many combined CA-based PRNGs which satisfy the maximal equidistribution, have almost maximal period and also pass all or majority of the empirical tests in all testbeds. Then we show that the performance of these light-weight combined linear maximal length CA-based PRNGs are comparable to the Mersenne Twister and other several existing linear generators. and even better than the Mersenne Twister in terms of equidistribution and speed. However, since we use time spacing, the computation time is increased. Therefore, future work may involve designing a CA-based PRNG that achieves period comparable to Mersenne Twister or WELL while also ensuring faster performance and maximal equidistribution.

\section*{Acknowledgments}
The authors are grateful to Prof. Sukanta Das and all the mentors of the Indian Summer School on Cellular Automata 2025 for their continuous support, guidance and valuable feedback that shaped this work. This work is partially supported by Visvesvaraya PhD Scheme, Department of Electronics and Information Technology, Ministry of Communication and IT, Govt. of India.

\end{document}